\documentclass[lettersize,journal]{IEEEtran}

\usepackage[T1]{fontenc}

\usepackage{bm}
\usepackage{amsthm}
\usepackage{amsmath}

\newtheorem{theorem}{Theorem}

\usepackage{makecell}

\usepackage{cite}

\usepackage[pdftex]{graphicx}

\ifCLASSINFOpdf

\else

\fi

\usepackage{diagbox}
\usepackage{xcolor}

\usepackage{amsmath}

\newtheorem{remark}{Remark} 

\usepackage[cmintegrals]{newtxmath}

\usepackage{algorithm}

\usepackage{algcompatible}

\usepackage{hyperref}
\usepackage{subfigure}

\newtheorem{assumption}{\textbf{Assumption}}

\begin{document}
%
\title{Decentralizing Coherent Joint Transmission Precoding via Fast ADMM with Deterministic Equivalents}
%
%
%

\author{Xinyu~Bian,~\IEEEmembership{Graduate Student Member,~IEEE,} Yuhao~Liu, Yizhou~Xu, Tianqi~Hou, Wenjie~Wang,~\IEEEmembership{Member,~IEEE,}
        Yuyi~Mao,~\IEEEmembership{Member,~IEEE,}
        and~Jun~Zhang,~\IEEEmembership{Fellow,~IEEE}
\thanks{X. Bian is with the Department of Electronic and Computer Engineering, the Hong Kong University of Science and Technology, Hong Kong, China, and also with the Theory Lab, Central Research Institute, 2012 Labs, Huawei Technologies Co., Ltd., Hong Kong, China (E-mail: xinyu.bian@connect.ust.hk). Y. Liu is with the Department of Mathematical Sciences, Tsinghua University, Beijing 100084, China (E-mail: yh-liu21@mails.tsinghua.edu.cn). Y. Xu is with International Center of Theoretical Physics, Trieste, I-34151, Italy (Email: yizhou.xu.cs@gmail.com). T. Hou and W. Wang are with the Theory Lab, Central Research Institute, 2012 Labs, Huawei Technologies Co., Ltd., Hong Kong, China (houtianqi2@huawei.com, wang.wenjie@huawei.com). Y. Mao is with the Department of Electrical and Electronic Engineering, the Hong Kong Polytechnic University, Hong Kong, China (E-mail: yuyi-eie.mao@polyu.edu.hk). J. Zhang is with the Department of Electronic and Computer Engineering, the Hong Kong University of Science and Technology, Hong Kong, China (E-mail: eejzhang@ust.hk) \emph{(Xinyu Bian and Yuhao Liu are co-first authors.) (Corresponding author: Tianqi Hou)}

Part of this work was accepted by the 2024 IEEE International Conference on Acoustics, Speech, and Signal Processing (ICASSP) \cite{yliu2024}.}
\thanks{}
\thanks{}}

%
%

\markboth{}%
{Shell \MakeLowercase{\textit{et al.}}: Bare Demo of IEEEtran.cls for IEEE Communications Society Journals}
%



\maketitle

\begin{abstract}
Inter-cell interference (ICI) suppression is critical for multi-cell multi-user networks. In this paper, we investigate advanced precoding techniques for coordinated multi-point (CoMP) with downlink coherent joint transmission, an effective approach for ICI suppression. Different from the centralized precoding schemes that require frequent information exchange among the cooperating base stations, we propose a decentralized scheme to minimize the total power consumption. In particular, based on the covariance matrices of global channel state information, we estimate the ICI bounds via the deterministic equivalents and decouple the original design problem into sub-problems, each of which can be solved in a decentralized manner. To solve the sub-problems at each base station, we develop a low-complexity solver based on the alternating direction method of multipliers (ADMM) in conjunction with the convex-concave procedure (CCCP). Simulation results demonstrate the effectiveness of our proposed decentralized precoding scheme, which achieves performance similar to the optimal centralized precoding scheme. Besides, our proposed ADMM solver can substantially reduce the computational complexity, while maintaining outstanding performance. 
\end{abstract}


%
\IEEEpeerreviewmaketitle

\section{Introduction \label{sectioni}}
The proliferation of advanced mobile applications has driven the dramatic increase of global mobile data traffic, which was forecasted to grow
$20 \sim 30$ percent per annum \cite{ericsson}. As a result, the demand for gigabits-per-second data rates and extremely high spectral efficiencies has increased significantly to support ultimate user experience \cite{5gppp}. Although network densification has been proposed as one promising solution for fifth generation (5G) mobile networks \cite{nbh2014}, it brings substantial inter-cell interference (ICI) that may severely deteriorate the communication performance. In particular, the cell-edge users in dense wireless networks may fail to be served by any surrounding base station (BS) due to severe ICI \cite{xyou2011}. Therefore, coordinated multi-point (CoMP) \cite{msa2010} has been proposed to mitigate ICI, which refers to a system where the transmissions and/or receptions at multiple, geographically separated BSs are dynamically coordinated. Specifically, in the downlink procedure of CoMP, user-centric coherent joint transmission (CJT) \cite{eda2013} is preferred, where a user equipment (UE) may be served by several BSs simultaneously so that its received signal power can be boosted. Besides, coordinated precoding is vital to multi-antenna CoMP for ICI suppression, where the global channel state information (CSI) is leveraged to exploit the spatial degrees of freedom of multi-antenna at BSs.

The coordinated precoders are typically optimized in a centralized manner, which requires the transmission of global CSI to the core network (CN). Moreover, the computation and distribution of the optimized precoding matrices have to be completed within in a transport time interval (TTI), i.e., $0.5$ ms. Nevertheless, most 5G cells deployed around the world are upgraded from the fourth generation (4G) mobile networks, which adopt the internet protocol-based radio access networks (IP-RAN) architecture \cite{3gpp2017}, and the round-trip communication latency between a BS and the CN is $2 \sim 4$ ms. This indicates that the existing IP-RAN architecture cannot support the optimization of coordinated precoding schemes in a centralized way. Although the cloud-RAN (C-RAN) \cite{yshi2015, mpeng2015} architecture is able to significantly reduce the transmission latency of CSI for centralized precoding, the cost of updating IP-RAN is fairly high since it requires erecting optical fronthauls and backhauls with high bandwidth and low latency \cite{dpli2018}. Therefore, it is imperative to investigate how to efficiently realize coordinated multi-cell precoding for user-centric CJT in the IP-RAN architecture.

\subsection{Related Works and Motivations}
As one of the most fundamental approaches to control and even completely eliminate the inter-user and inter-cell interference, various precoding schemes have been proposed for user-centric CJT. In particular, capacity-achieving dirty-paper coding (DPC) schemes were proposed for multi-antenna systems in \cite{hzha2004,wyu2007}. Considering the implementation complexity of DPC, cost-effective linear precoding schemes such as zero forcing (ZF) \cite{ssh2008,rzhang2010} have been investigated for cooperative multi-cell systems, where the inter-user interference is suppressed by jointly optimizing the precoders among coordinated BSs. Besides, in order to reduce the system complexity while maintaining
the benefits of CJT, an overlapping clustered coordination structure \cite{sve2007} for networked MIMO systems was considered in \cite{jzhang2009}, based on which, a ZF precoding scheme was developed. Another popular formulation for coordinated beamforming is to optimize the system performance measured by some utility function. For example, \cite{qshi2011} considered a sum rate maximization problem with a total transmit power constraint of all BSs, for which, an optimal iterative weighted minimum mean square error (WMMSE) algorithm was developed. Similarly, the transmit power minimization problem with per-user signal-to-interference-plus-noise ratio (SINR) constraints was investigated in \cite{hda2010}, which was optimally solved by leveraging the uplink-downlink duality (UDD). Although the two problems have been proved to be dual problems with the same optimal solution \cite{ebj2014}, neither of them was suitable for user-centric CJT. Thus, authors of \cite{zwu2018} renovated the WMMSE algorithm for user-centric CJT by further respecting the power constraints of BS antennas. 

However, all these methods were designed for centralized coordinated beamforming, which requires to transmit the instantaneous global CSI and optimized precoding matrices through fronthauls and backhauls, resulting in significant overhead and latency. Therefore, a distributed precoding scheme for user-centric CJT is needed for practical implementation to limit the amount of exchanged information for coordinated beamforming. To reduce the backhaul overhead, \cite{sh2013} investigated joint precoding and data compression for C-RAN. Besides, the uplink pilot signals were reused for downlink transmissions in \cite{jka2017} to avoid transmitting downlink CSI through the backhauls. Although these methods attempted to perform decentralized precoding with limited information exchange between the CN and BSs, they are still not applicable in the IP-RAN architecture due to the per-TTI information exchange.


\subsection{Contributions}
In this paper, we endeavor to develop a low-complexity decentralized precoding scheme for user-centric CJT in the IP-RAN architecture, where the cooperating BSs locally determine near-optimal precoders with minimal information exchange. Our main contributions are summarized below.

\begin{itemize}
    \item Considering the significant latency of information exchange in the IP-RAN architecture, we propose a decentralized precoding scheme for user-centric CJT, which only requires the covariance matrices of global CSI and avoids real-time information exchange in each transmission block. In particular, the transmit power minimization problem is investigated. Based on the covariance matrices of global CSI, we leverage deterministic equivalents (DE) \cite{RMT2011} from random matrix theory (RMT) to estimate the inter-cell interference, which decomposes the original design problem for individual BSs.
\end{itemize}

\begin{itemize}
    \item We show that the beamforming design problem at each BS is a non-convex quadratically constrained quadratic programming (QCQP) problem, which is in general NP-hard. Although these problems can be transformed to convex semi-definite programming (SDP) or second order cone programming (SOCP) problems, the solution complexity still hinders practical implementation at BSs. Therefore, a fast algorithm is proposed by adopting the alternating direction method of multipliers (ADMM) in conjunction with the convex-concave procedure (CCCP).
\end{itemize}

\begin{itemize}
    \item Simulation results show that the proposed precoding scheme with a blackbox SOCP solver achieves near-optimal sum-rate performance, which secures $8\% \sim 57\%$ improvements compared with the ZF-based centralized precoding scheme. Besides, the proposed fast ADMM-based solver only bears less than $5\%$ sum-rate loss compared with the SOCP solver, while significantly reduces $79\%$ of the computational overhead.
\end{itemize}

\subsection{Organization}
The rest of this paper is organized as follows. We introduce the user-centric CJT system and formulate the coordinated beamforming problem in Section \ref{sectionii}. In Section \ref{sectioniii}, we propose a method to estimate the inter-cell interference by adopting DE so that the problem can be decoupled at each BS. A low-complexity ADMM-based solver is developed in Section \ref{sectioniv} to obtain the precoding matrices at each BS. Simulation results are presented in Section \ref{sectionv} and Section \ref{sectionvi} concludes this paper.

\subsection{Notations} 
We use lower-case letters, bold-face lower-case letters, bold-face upper-case letters, and math calligraphy letters to denote scalars, vectors, matrices, and sets, respectively. The transpose and conjugate transpose of a matrix $\mathbf{M}$ are denoted as $\mathbf{M}^{T}$ and $\mathbf{M}^{H}$, respectively. Besides, we denote the complex Gaussian distribution with mean $\mu$ and variance $\sigma^2$ as $\mathcal{C} \mathcal{N}(\mu, \sigma^2)$. In addition, $\Re\{\cdot\}$ returns the real part of a complex number. Finally, we use $a.s.$ to represent almost surely convergence. Given a sequence of random variables $\{ X_ n\}$ and a random variable $X$, we assert that $X_n \to X$ almost surely (a.s.) if $P(\lim_{n \to \infty} X_n \neq X) = 0$.

\section{System Model and Problem Formulation \label{sectionii}}
\subsection{System Model}
We consider the downlink CJT in a multi-cell multi-user MIMO system as shown in Fig. \ref{cjtsystem}, where $N_c$ UEs are served by $N_B$ BSs. Each BS is equipped with $N_T$ transmit antennas, and every UE has a single receive antenna. The ensemble of UEs is defined as $\mathcal{U} \triangleq \{1, 2, \ldots, N_c\}$, and the set of BSs is denoted as $\mathcal{T} \triangleq \{1, 2, \ldots, N_B\}$. The subset $\mathcal{U}_p \subseteq \mathcal{U}$ with cardinality $U_p$ represents a distinct subcollection of UEs associated with BS $p$. For signal quality enhancement, a UE indexed as $i$ can benefit from the collaboration of a set of BSs, denoted as $\mathcal{T}_i \subseteq \mathcal{T}$ with cardinality $T_i$. Let $\mathbf{h}_{ip} \in \mathbb{C}^N$ represent the channel vector from BS $p$ to UE $i$, and let $\mathbf{w}_{ip} \in \mathbb{C}^N$ symbolize the precoding vector adopted by UE $i$ at the intended BS $p$. Specially, the channel vector from BS $p$ to UE $i$ is expressed in the form of $\mathbf{h}_{ip} = \mathbf{\Theta}_{ip}^{1/2} \mathbf{z}_{ip}$, where $\mathbf{z}_{ip}$ represents the small-scale fading and consists of independent and identically distributed (i.i.d.) complex Gaussian entries with zero mean and unit variance. Matrix $\mathbf{\Theta}_{ip} \in \mathbb{C}^{N_T \times N_T}$ is the covariance matrix of $\mathbf{h}_{ip}$, encompassing the impact of pathloss. The received signal at UE $i$, denoted as $y_i \in \mathbb{C}$, comprises the desired signal, intra-cell interference, and inter-cell interference, which is expressed as follows:
\begin{equation} \label{eq:model}
\begin{aligned}
     y_i   &= \sum_{p \in \mathcal{T}_i} \mathbf{h}_{ip}^H \mathbf{w}_{ip} s_i + \sum_{q \in \mathcal{T}_i} \mathbf{h}_{iq}^H \sum_{j \in \mathcal{U}_q \setminus \{i\}} \mathbf{w}_{jq} s_j \\
          &+ \sum_{q \in \mathcal{T} \setminus \mathcal{T}_i} \mathbf{h}_{iq}^H \sum_{j \in \mathcal{U}_q } \mathbf{w}_{jq} s_j + n_i.
\end{aligned}
\end{equation}
\begin{figure}[t]
\centering
\includegraphics[width=3.4in]{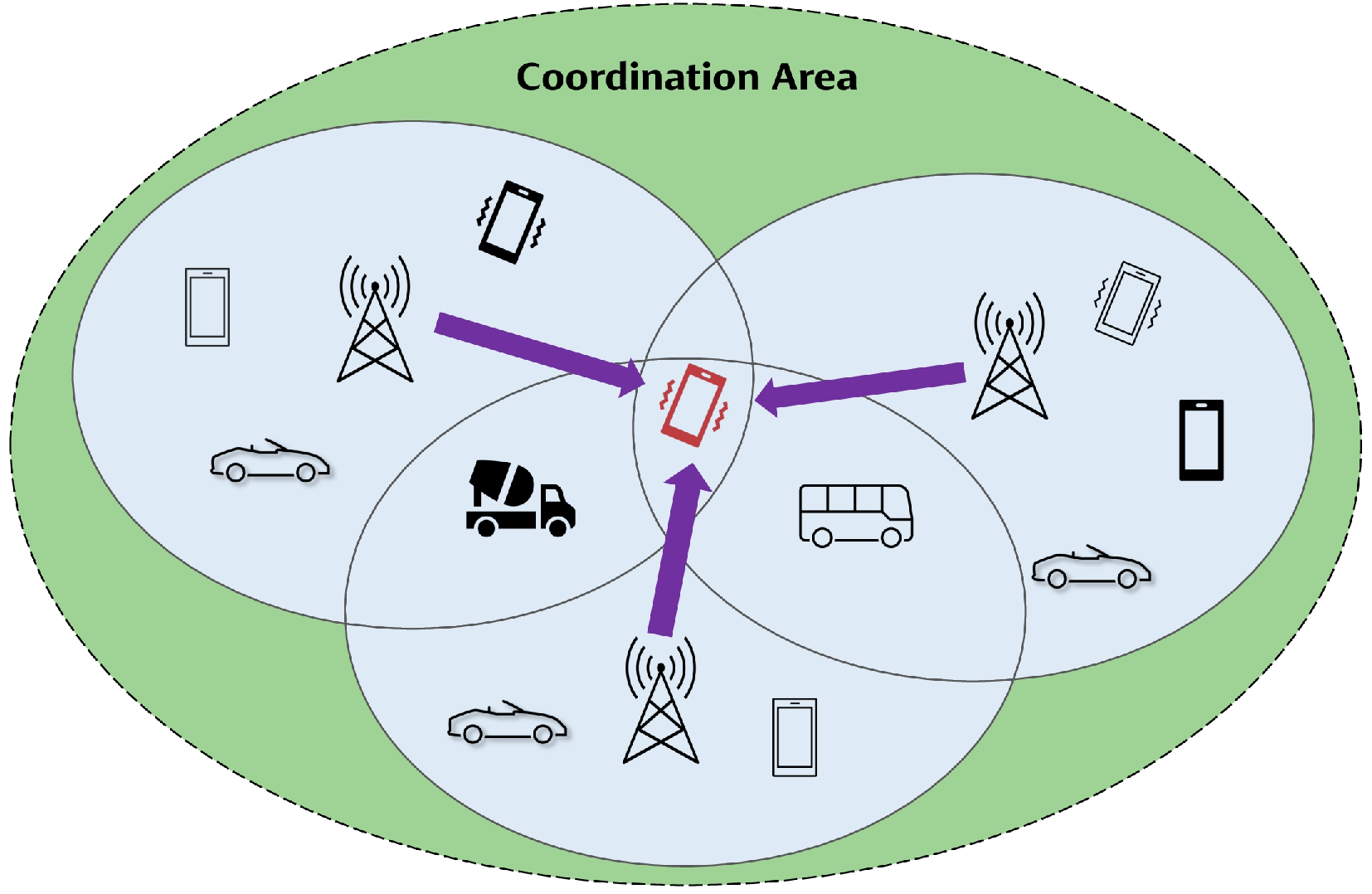}
\caption{A downlink CJT system, where the cell-edge UEs are served by multiple coordinating BSs simutaneously.}
\label{cjtsystem}
\end{figure}

\noindent Here $s_i$ represents the data symbol directed towards UE $i$ with zero mean and unit variance, and $n_i \sim \mathcal{C}\mathcal{N}(0,\sigma^2)$ stands for the additive white Gaussian noise. Based on the received signal model, the signal-to-interference-plus-noise ratio (SINR) at UE $i$ is given as follows:
\begin{equation}  \label{eq:originalSINR}
  \Gamma_{i} = \frac{\Big|\sum_{p\in\mathcal{T}_{i}}\mathbf{h}_{ip}^H \mathbf{w}_{ip}\Big|^2}{\displaystyle \sum_{j \in \mathcal{U} \setminus \{i\}} \Big|\sum_{q \in \mathcal{T}_{j}} \mathbf{h}_{iq}^H \mathbf{w}_{jq}\Big|^2 + \sigma^2}.
\end{equation}
The above SINR expression is entangled among the coordinated BSs, which prevents decentralized precoding. To facilitate the precoder design, we define the approximate SINR at UE $i$ concerning BS $p$ as follows:
\begin{equation}  \label{eq:SINR}
  \Gamma_{ip} = \frac{|\mathbf{h}_{ip}^H \mathbf{w}_{ip}|^2}{\displaystyle \sum_{q \in \mathcal{T}_i} \sum_{j \in \mathcal{U}_{q} \setminus \{i\}} |\mathbf{h}_{iq}^H \mathbf{w}_{jq}|^2 + \sum_{q \in \mathcal{T} \setminus \mathcal{T}_i} \sum_{j \in \mathcal{U}_{q} } |\mathbf{h}_{iq}^H \mathbf{w}_{jq}|^2+ \sigma^2}.
\end{equation}
However, the original SINR formulation in (\ref{eq:originalSINR}) is still used for performance evaluation for fair comparison, i.e., calculating the sum rate $R_s=\sum_{i \in \mathcal{U}}\log(1+\Gamma_i)$.

\subsection{Problem Formulation}
In the CoMP downlink, the BSs formulate precoders with the collective objective of minimizing the overall transmit power. The optimization is subject to the stipulated minimum SINR conditions for each BS with respect to its corresponding UEs, denoted as $\gamma_{ip}$, $\forall p$, $\forall i \in \mathcal{U}_p$. The formulation of the problem is expressed as follows:
\begin{equation}  \label{eq:opt1}
  \begin{aligned}
    & \min_{\{\mathbf{w}_{ip}\}}  \sum_{p \in \mathcal{T}} \sum_{i \in \mathcal{U}_p} \Vert \mathbf{w}_{ip} \Vert_2^2 \\
    &\ \ \text{s.t.} \quad \Gamma_{ip} \geq \gamma_{ip}, \quad \forall p \in \mathcal{T}, \forall i \in \mathcal{U}_p.
  \end{aligned}
\end{equation}
Let $\tau_{iq}$ and $\epsilon_{iq}$ denote the bounds of two inter-cell interference terms in (\ref{eq:SINR}), i.e., the interference from BS $q$ that serves UE $i$ or not. The 
optimization problem in (\ref{eq:opt1}) can thus be relaxed as follows:
\begin{subequations} \label{eq:opt2}
\begin{align}  
      &\min_{\{\mathbf{w}_{ip}, \epsilon_{iq}, \tau_{iq} \}}  
      \sum_{p \in \mathcal{T}} \sum_{i \in \mathcal{U}_p} \Vert \mathbf{w}_{ip} \Vert_2^2 \label{eq:opt21} \\
      &\ \text{s.t.}\ \frac{|\mathbf{h}_{ip}^H \mathbf{w}_{ip}|^2}{\displaystyle \sum_{j \in \mathcal{U}_p \setminus \{i\}}|\mathbf{h}_{ip}^H \mathbf{w}_{jp}|^2\! +\!\sum_{q \in \mathcal{T}_i \setminus \{p\}} \tau_{iq}\! +\! \sum_{q \in \mathcal{T} \setminus  \mathcal{T}_i} \epsilon_{iq} + \sigma^2} \geq \gamma_{ip},\nonumber \\
    &\quad \quad \quad \quad \quad \quad \quad \quad \quad \quad \quad \quad \quad \quad \quad \quad \quad \forall p \in \mathcal{T}, \forall i \in \mathcal{U}_p, \label{eq:opt22}\\
      &\quad \quad \sum_{j \in \mathcal{U}_q \setminus \{i\} } |\mathbf{h}_{iq}^H \mathbf{w}_{jq}|^2 \leq \tau_{iq}, \forall q,\ \forall i \in \mathcal{U}_q, \label{eq:opt23}\\
      &\quad \quad \sum_{j \in \mathcal{U}_q} |\mathbf{h}_{iq}^H \mathbf{w}_{jq}|^2 \leq \epsilon_{iq}, \forall q,\ \forall i \in \mathcal{U} \setminus \mathcal{U}_{q}. \label{eq:opt24}
\end{align}
\end{subequations}
As shown in Appendix \ref{sec:dual}, Problems (\ref{eq:opt1}) and (\ref{eq:opt2}) share identical optimal solutions $\{\mathbf{w}_{ip}\}$, at which, constraints specified in (\ref{eq:opt23}) and (\ref{eq:opt24}) are satisfied with equality. Therefore, once $\{\tau_{iq}\}$'s and $\{ \epsilon_{iq}\}$'s are determined, we are able to decouple Problem (\ref{eq:opt2}) for the $p$-th BS, $\forall p \in \mathcal{T}$ as follows:
\begin{subequations} \label{eq:opt3}
\begin{align}  
      \min_{\{\mathbf{w}_{ip} \}}  
    &\sum_{i \in \mathcal{U}_p} \Vert \mathbf{w}_{ip} \Vert_2^2 \label{eq:opt31} \\
      \text{s.t.}\ &\frac{|\mathbf{h}_{ip}^H \mathbf{w}_{ip}|^2}{\displaystyle \sum_{j \in \mathcal{U}_p \setminus \{i\}} |\mathbf{h}_{ip}^H \mathbf{w}_{jp}|^2\! +\! \sum_{q \in \mathcal{T}_i \setminus \{p\}} \tau_{iq}\! +\! \sum_{q \in \mathcal{T} \setminus  \mathcal{T}_i} \epsilon_{iq} + \sigma^2} \geq \gamma_{ip},\nonumber \\
    &\quad \quad \quad \quad \quad \quad \quad \quad \quad \quad \quad \quad \quad \quad \quad \quad \quad \quad \quad \forall i \in \mathcal{U}_p, \label{eq:opt32}\\
      &\ \sum_{j \in \mathcal{U}_p \setminus \{i\}} |\mathbf{h}_{ip}^H \mathbf{w}_{jp}|^2 \leq \tau_{ip},\ \forall i \in \mathcal{U}_p, \label{eq:opt33}\\
      &\ \sum_{j \in \mathcal{U}_p} |\mathbf{h}_{ip}^H \mathbf{w}_{jp}|^2 \leq \epsilon_{ip},\ \forall i \in \mathcal{U} \setminus \mathcal{U}_{p}. \label{eq:opt34}
\end{align}
\end{subequations}

The optimization variables for each sub-problem are exclusively the precoding vectors of the respective BS, and the constraints are solely related to the local CSI. Thus, each of these problems can be reformulated to an equivalent convex problems \cite{zql2006}, such as a second order cone programming (SOCP) and semidefinite programming (SDP) problem, which can be solved with many off-the-shelf solvers. For example, Problem (\ref{eq:opt3}) can be transformed to the following SOCP problem:
\begin{align}  \label{SOCP}
  \begin{aligned}
    & \min_{S,\{\mathbf{w}_{ip}\}} \ \ S \\
    &\ \ \text{s.t.}\ \ {\left[\begin{array}{c}
\sqrt{1+\frac{1}{\gamma_{ip}}} \mathbf{h}_{ip}^{H} \mathbf{w}_{ip} \\
\text{[}\mathbf{h}_{ip}^{H}\mathbf{w}_{ap}, \cdots, \mathbf{h}_{ip}^{H}\mathbf{w}_{bp}\text{]}^{T} \\
\text{[}\sqrt{\tau_{ic}},\cdots,\sqrt{\tau_{id}}\text{]}^{T}\\
\text{[}\sqrt{\epsilon_{ie}},\cdots,\sqrt{\epsilon_{if}},\sigma\text{]}^{T}\\
\end{array}\right] \succeq_{\mathrm{SOC}} 0, \forall i \in \mathcal{U}_{p}}, \\
&\ \ \ \ \ \! \ \ \ {\left[\begin{array}{c}
\sqrt{\tau_{ip}} \\
\text{[}\mathbf{h}_{ip}^{H}\mathbf{w}_{gp},\cdots,\mathbf{h}_{ip}^{H}\mathbf{w}_{hp}\text{]}^{T}
\end{array}\right] \succeq_{\mathrm{SOC}} 0, \forall i \in \mathcal{U}_p}, \\
&\ \ \ \ \ \! \ \ \ {\left[\begin{array}{c}
\sqrt{\epsilon_{ip}} \\
\text{[}\mathbf{h}_{ip}^{H}\mathbf{w}_{ap},\cdots,\mathbf{h}_{ip}^{H}\mathbf{w}_{bp}\text{]}^{T}
\end{array}\right] \succeq_{\mathrm{SOC}} 0, \forall i \in \mathcal{U} \setminus \mathcal{U}_p}, \\
&\ \ \ \ \ \ \! \ \ {\left[\begin{array}{c}
S \\
\text{[}\mathbf{w}_{ap}^{T},\cdots,\mathbf{w}_{bp}^{T}\text{]}^{T}
\end{array}\right] \succeq_{\text {SOC }} 0},
\end{aligned}
\end{align}
where $[x,\mathbf{y}^{T}]^{T}\succeq_{\text {SOC }} 0$ denotes $x\geq||\mathbf{y}||_{2}$, $a,b\in \mathcal{U}_{p}$, $c,d \in \mathcal{T}_{i}\setminus \{p\}$, $e,f \in \mathcal{T}\setminus \mathcal{T}_{p}$, and $g,h \in \mathcal{U}_{p}\setminus \{i\}$. However, $\{\tau_{iq}\}$'s and $\{\epsilon_{iq}\}$'s need to be judiciously decided as their values greatly affect the performance of the resulting decentralized precoding scheme. Therefore, determining the interference bounds with minimal information exchange is crucial, which will be investigated in the next section.

\section{Decentralized Precoding for User-centric CJT \label{sectioniii}}
As discussed in the preceding section, the transmit power minimization problem at each BS can be solved distributively once the inter-cell interference bounds are determined. However, the solution may perform far worse than the global optimal solution if values of the interference bounds are set randomly. Therefore, a pivot problem arises: How to determine the interference bounds with minimal information exchange such that the decentralized precoding scheme can achieve the near-optimal performance? To address this issue, we propose a novel approach by employing DE, which is a powerful tool capable of computing specific functions of large random matrices. 

\subsection{Centralized Calculation of Interference Bounds} \label{sec:central}
First, without considering the cost of information exchange, we present a centralized method to determine the interference bounds using the UDD-based centralized optimal precoding scheme \cite{hda2010}, which was originally proposed for multi-cell non-CJT systems. 
With the Lagrangian analysis in Appendix \ref{sec:dual}, for UE $i$ served by BS $p$, the optimal inter-cell interference bounds can be expressed as a function of global CSI as follows:
\begin{equation}  \label{eq:tau}
  \tau_{iq} = \sum_{j \in \mathcal{U}_q \setminus \{i\}} \frac{1}{N_T}  \delta_{jq} |\mathbf{h}_{iq}^H \hat{\mathbf{w}}_{jq}|^2, \quad \forall q, \forall i \in \mathcal{U}_q,
\end{equation}
\begin{equation}  \label{eq:eps}
  \epsilon_{iq} = \sum_{j \in \mathcal{U}_q } \frac{1}{N_T} \delta_{jq} |\mathbf{h}_{iq}^H \hat{\mathbf{w}}_{jq}|^2, \quad \forall q, \forall i \notin \mathcal{U}_{q}.
\end{equation}
In particular, $\delta_{ip}$ serves as the scaling factor relating $\mathbf{w}_{ip}$ with $\hat{\mathbf{w}}_{ip}$, i.e., the optimal precoding vector can be expressed as $\mathbf{w}_{ip} = \sqrt{\frac{\delta_{ip}}{N_T}} \hat{\mathbf{w}}_{ip}$ with 
\begin{equation} \label{eq:hatw}
  \hat{\mathbf{w}}_{ip} = (N_T \mathbf{I} +  \sum_{j \in \mathcal{U} \setminus \{i\}} \sum_{q \in \mathcal{T}_j} \lambda_{jq}^* \mathbf{h}_{jp} \mathbf{h}_{jp}^H)^{-1} \mathbf{h}_{ip},
\end{equation}
where the optimal Lagrangian multipliers $\{\lambda_{ip}^{*}\}$ are the unique solution of the following equation set that can be obtained by fixed-point iteration:
\begin{equation} \label{eq:lamda}
    \begin{aligned}
        \lambda_{ip} &= \frac{\gamma_{ip}}{\mathbf{h}_{ip}^H(N_T \mathbf{I} +  \sum_{j \in \mathcal{U} \setminus \{i\}} \sum_{q \in \mathcal{T}_j} \lambda_{jq} \mathbf{h}_{jp} \mathbf{h}_{jp}^H)^{-1} \mathbf{h}_{ip}}, \\
   &\qquad \qquad \qquad \qquad \qquad \qquad \qquad \quad \forall p, \forall i \in \mathcal{U}_p.
    \end{aligned}
\end{equation} 
Besides, the scaling factors $\{\delta_{ip}\}$'s can be determined such that the SINR constraints in (\ref{eq:opt22}) are achieved with equality, i.e. $\boldsymbol{\delta} = N_T \sigma^2 \mathbf{F}^{-1} \boldsymbol{1}$ with $\boldsymbol{\delta} = [\boldsymbol{\delta}_1, \boldsymbol{\delta}_2, \cdots, \boldsymbol{\delta}_p]^T$ and $\boldsymbol{\delta}_p = \{{\delta}_{ip}\}_{i \in \mathcal{U}_p}$ as shown in Appendix \ref{sec:solution}. Here,
$\boldsymbol{1}$ is the $(\sum_{p \in \mathcal{T}} U_p) $-dimensional all ones-vector and $\mathbf{F}$ is expressed as follows:
\begin{equation} \label{eq:F1}
  \mathbf{F}=\left[\begin{array}{cccc}
      \mathbf{F}^{11} & \mathbf{F}^{12} & \ldots & \mathbf{F}^{1 N_B} \\
      \mathbf{F}^{21} & \mathbf{F}^{22} & \ldots & \mathbf{F}^{2 N_B} \\
      \cdot & & & \\
      \cdot & & & \\
      \cdot & & & \\
      \mathbf{F}^{N_B 1} & \mathbf{F}^{N_B 2} & \ldots & \mathbf{F}^{N_B N_B}
      \end{array}\right],
\end{equation}
where the $(i,j)^{\text{th}}$ element of the so-called coupling matrix $\mathbf{F}^{pq}$ is defined as
\begin{equation} \label{eq:F2}
    {F}_{{ij}}^{{pq}}=\left\{\begin{array}{ll}
        \frac{1}{\gamma_{ip}}\left|\hat{\mathbf{w}}_{ip}^{H} \mathbf{h}_{ip}\right|^{2}, & \text { if } q=p \text { and } j=i, \\
        -\left|\hat{\mathbf{w}}_{jq}^{H} \mathbf{h}_{iq}\right|^{2}, & \text { if } q \in \mathcal{T} \text { and } j \in \mathcal{U}_q \setminus \{i\},\\
        0, & \text { else } (q \in \mathcal{T} \setminus \{p\}, j = i).
        \end{array}\right.
\end{equation}

Such a method to calculate the interference bounds requires the exchange of instantaneous channel vectors among BSs, thereby incurring substantial communication expenses. As a result, by leveraging DE techniques from random matrix theory, we delineate the asymptotic behaviour of interference bounds with statistical channel information, circumventing the need for precise knowledge of individual channel vectors. In the following, we present a collection of robust approximations for the interference bounds by relying solely on the covariance matrices of global CSI. 

\subsection{Approximation of Interference Bounds}  
\label{sec:LSA}
For effective utilization of DE techniques, it is necessary to make the following widely adopted assumptions to properly depict the growing rate of system dimensions.
\begin{assumption}   \label{eq:ass2}
    $0<\lim \limits_{ N_T \to \infty} {\inf} \frac{N_c}{N_T} \leq \lim \limits_{ N_T \to \infty} {\sup} \frac{N_c}{N_T} < \infty$.
\end{assumption}
\begin{assumption}   \label{eq:ass3}
   $\lim \limits_{ N_T \to \infty} \sup \max _{\forall i, p}\left\{\left\|\boldsymbol{\Theta}_{ip}\right\|\right\}<\infty$. 
\end{assumption}
Firstly, we compute the DE of the optimal Lagrangian multipliers $\{\lambda_{ip}^{*}\}$. By applying analogous analytical techniques involving the rank-1 permutation lemma \cite{SJW1995} and the trace lemma \cite{bai1998}, we deduce the following result.
\begin{theorem} \label{eq:th1}
   Let Assumptions \ref{eq:ass2} and \ref{eq:ass3} hold. We have $\max_{i,p} |\lambda_{ip}^* - \bar{\lambda}_{ip}| \to 0$ almost surely where
   \begin{equation} \label{eq:delamda}
     \bar{\lambda}_{ip} = \frac{\gamma_{ip}}{\bar{m}_{ip}}, \quad \forall p \in \mathcal{T}, i \in \mathcal{U}_p,
   \end{equation}
    and $\bar{m}_{ip}$ is given by the unique non-negative solution to the following system equations:
   \begin{equation}  \label{eq:dem}
        \begin{aligned}
            \bar{m}_{ip} = \operatorname{Tr}\left(\boldsymbol{\Theta}_{ip} \left(\sum_{j \in \mathcal{U}} \frac{(\sum_{q \in \mathcal{T}_j } \bar{\lambda}_{jq} ) \boldsymbol{\Theta}_{jp}}{1+ (\sum_{q \in \mathcal{T}_j } \bar{\lambda}_{jq} ) \bar{m}_{jp}} + N_T \mathbf{I}\right)^{-1}\right)&, \\
            \quad \quad \quad \quad \forall p \in \mathcal{T}, i \in \mathcal{U}_p &. 
        \end{aligned} 
   \end{equation}
\end{theorem}
\begin{proof}
    The key to the proof is to express the asymptotic value of the denominator in (\ref{eq:lamda}) with covariance matrix $\mathbf{\Theta}_{ip}$. First, we employ the trace lemma to express the denominator as $\operatorname{Tr} [\boldsymbol{\Theta}_{ip} (N_T \mathbf{I} +  \sum_{j \in \mathcal{U} \setminus \{i\}} \sum_{q \in \mathcal{T}_j} \lambda_{jq} \mathbf{h}_{jp} \mathbf{h}_{jp}^H)^{-1}]$, thereby eliminating $\textbf{h}_{ip}$. Second, we add the term $\lambda_{iq} \textbf{h}_{iq} \textbf{h}_{iq}^H$ to the matrix being inverted, yielding $\operatorname{Tr} [\boldsymbol{\Theta}_{ip} (N_T \mathbf{I} +  \sum_{j \in \mathcal{U}} \sum_{q \in \mathcal{T}_j} \lambda_{jq} \mathbf{h}_{jp} \mathbf{h}_{jp}^H)^{-1}]$. This operation does not change the asymptotic value according to the rank-one lemma. Finally, by employing the DE to eliminate $\textbf{h}_{jp}^H \textbf{h}_{jp}$, we compute  the asymptotic value of $\operatorname{Tr} [\boldsymbol{\Theta}_{ip} (N_T \mathbf{I} +  \sum_{j \in \mathcal{U}} \sum_{q \in \mathcal{T}_j} \lambda_{jq} \mathbf{h}_{jp} \mathbf{h}_{jp}^H)^{-1}]$. The detailed proof is available in Appendix \ref{sec:proof1}.

\end{proof}
Theorem \ref{eq:th1} firmly establishes that in the large-system limit, the asymptotic value  $\bar{\lambda}_{ip}$ is independent of the instantaneous channel vectors. Instead, it becomes wholly computable through utilizing the covariance matrix of global CSI. As such, we are able to derive the DE of individual entries of the coupling matrix $\mathbf{F}$ and scaling factors $\boldsymbol{\delta}$. By substituting $\bar{\lambda}_{ip}$ into $\hat{\mathbf{w}}_{ip}$, we have $\hat{\bar{\mathbf{w}}}_{ip}= (N_T \mathbf{I} +  \sum_{j \in \mathcal{U} \setminus \{i\}} 
\sum_{q \in \mathcal{T}_j} \bar{\lambda}_{jq} \mathbf{h}_{jp} \mathbf{h}_{jp}^H)^{-1} \mathbf{h}_{ip}$, which is an asymptotic expression of $\hat{\mathbf{w}}_{ip}$. By further substituting $\{\hat{\bar{\mathbf{w}}}_{ip}\}$ into the coupling matrix in (\ref{eq:F2}), we obtain the DE of its entries in the following theorem.
\begin{theorem} \label{eq:th2}
Let Assumptions \ref{eq:ass2} and \ref{eq:ass3} hold. Given $\bar{\lambda}_{ip}$ and 
$\bar{m}_{ip}$, $\forall p \in \mathcal{T}, \forall i \in \mathcal{U}_p$ defined in Theorem \ref{eq:th1}, we have $|F^{pq}_{ij} - \bar{F}^{pq}_{ij}| \to 0$
almost surely with 
\begin{equation}  \label{eq:deF}
    \bar{F}_{i j}^{pq}=\left\{\begin{array}{ll}
        \frac{1}{\gamma_{i p}} \bar{m}_{ip}^2, & \text { if } q=p \text { and } j=i, \\
        -\frac{1}{N_T} \frac{\bar{m}_{j,i,q}^{\prime}}{(1 + (\sum_{r \in \mathcal{T}_i} \bar{\lambda}_{ir}) \bar{m}_{iq})^2}, & \text { if } q \in \mathcal{T} \text { and } j \in \mathcal{U}_{q} \backslash i, \\
        0, & \text { else }(q \in \mathcal{T} \backslash p, j=i),
        \end{array}\right.
\end{equation}
where we have $[\bar{m}_{1,i,q}^{\prime},\bar{m}_{2,i,q}^{\prime},\cdots,\bar{m}_{N_c,i,q}^{\prime}]^T= (\mathbf{I}_{N_c} - \mathbf{L}_q)^{-1} \mathbf{u}_{iq}$, $\forall 
p \in \mathcal{T}, \forall i \in \mathcal{U}_p$, with entries of $\mathbf{L}_q$ being
\begin{equation}
  [L_q]_{hl} = \frac{1}{N_T^2} \frac{\operatorname{Tr} \left(\boldsymbol{\Theta}_{hq} \mathbf{T}_q (\sum_{r \in \mathcal{T}_l} \bar{\lambda}_{rl})^2 \boldsymbol{\Theta}_{lq} \mathbf{T}_q\right)}{[1 + (\sum_{r \in \mathcal{T}_l} \bar{\lambda}_{rl}) \bar{m}_{lq}]^2}.
\end{equation}
Besides,
\begin{equation}
\begin{array}{r}
\mathbf{u}_{i q}=\left[\frac{1}{N_T} \operatorname{Tr}\left(\boldsymbol{\Theta}_{1 q} \mathbf{T}_{q} \boldsymbol{\Theta}_{i q} \mathbf{T}_{q}\right), \frac{1}{N_T} \operatorname{Tr}\left(\boldsymbol{\Theta}_{2 q} \mathbf{T}_{q} \boldsymbol{\Theta}_{i q} \mathbf{T}_{q}\right), \cdots\right. \\
\left.\frac{1}{N_T} \operatorname{Tr}\left(\boldsymbol{\Theta}_{N_{c} q} \mathbf{T}_{q} \boldsymbol{\Theta}_{i q} \mathbf{T}_{q}\right)\right]^{T}
\end{array}
\end{equation}
with $\mathbf{T}_q=\left(\frac{1}{N_T} \sum_{k \in \mathcal{U}} \frac{(\sum_{r \in \mathcal{T}_k} \bar{\lambda}_{kr}) \boldsymbol{\Theta}_{kq}}{1 + (\sum_{r \in \mathcal{T}_k } \bar{\lambda}_{kr}) \bar{m}_{kq}} + \mathbf{I}\right)^{-1}$.
The DE of $\{F_{ij}^{pq}\}$, i.e., $\{\bar{F}_{ij}^{pq}\}$, can be employed to compute the asymptotically optimal scaling factors as $\bar{\boldsymbol{\delta}} = N_T \sigma^2 \bar{\boldsymbol{F}}^{-1} \boldsymbol{1}$.  
\end{theorem}
\begin{proof}
    The proof is available in Appendix \ref{sec:proof2}. 
\end{proof}
Based on the theorems established herein, we derive DE approximations for all scalar parameters associated with the interference bounds. Next, by substituting the obtained DE expressions for $\boldsymbol{\delta}$ and $F_{ij}^{pq}$ into (\ref{eq:tau}) and (\ref{eq:eps}), we approximate the interference bounds as follows:
\begin{equation} \label{eq:btau}
    \bar{\tau}_{iq} = - \frac{1}{N_T} \sum_{j \in \mathcal{U}_q \setminus \{i\}}  \bar{\delta}_{jq} \bar{F}^{qq}_{ij},  \forall q, \forall i \in \mathcal{U}_q,
\end{equation}
\begin{equation} \label{eq:beps}
    \bar{\epsilon}_{iq} = -\frac{1}{N_T} \sum_{j \in \mathcal{U}_q}  \bar{\delta}_{jq} \bar{F}^{qq}_{ij}  ,\forall q, \forall i \notin \mathcal{U}_q.
\end{equation}
It is notable that the above approximations only require the channel covariance matrices $\{\boldsymbol{\Theta}_{ip}\}$ at BSs, which can be categorized as slowly changing variables and thus can be exchanged only once for a long time. We outline the procedures for calculating $\{\bar{\tau}_{iq}\}$ and $\{\bar{\epsilon}_{iq}\}$ in Algorithm \ref{alg:alg1}.

\begin{algorithm} 
\caption{Calculation of DE of the interference bounds \label{alg:alg1}}
\begin{algorithmic}[1]
\STATE \textbf{REQUIRE:} Covariance matrices $\{\boldsymbol{\Theta}_{ip}\}$, target SINR $\{\gamma_{ip}\}$, noise variance $\sigma^2$.

\STATE Calculate the DE of $\{\lambda_{ip}^*\}$ by solving the fixed equations shown in (\ref{eq:delamda}) and (\ref{eq:dem}). 
\STATE Calculate the DE of $\boldsymbol{F}$ and $\boldsymbol{\delta}$ according to Theorem \ref{eq:th2}.
\STATE Calculate the DE of the interference bounds $\{\tau_{iq}\}$ and $\{\epsilon_{iq}\}$ according to (\ref{eq:btau}) and (\ref{eq:beps}).

\STATE \textbf{RETURN:} Interference approximations $\{ \bar{\tau}_{iq} \}$, $\{\bar{\epsilon}_{iq} \}$
\end{algorithmic}
\end{algorithm}

The implementation of the proposed decentralized precoding scheme in the IP-RAN architecture is shown in Fig \ref{IPRAN}, where each cell is equipped with a switch communicating with that of CN. In order to show that our proposed decentralized precoding scheme is effective in reducing the requirements of fronthaul and backhaul bandwidth and latency, we consider the following example: Suppose each BS has $64$ antennas and each user is equipped with $4$ antennas. The system bandwidth is $20$ MHz, and the subcarrier spacing is $30$ KHz. Thus, there exist $50$ resource blocks (RBs), i.e., $600$ resource elements (REs). Due to the channel correlation in adjacent RBs, we assume every $4$ RBs share the same channel matrix. Besides, each single-precision complex number requires $64$-bit data to represent. As a result, in the centralized precoding scheme, the backhaul needs to transmit $208$-Kbit CSI data for each UE, and the fronthual transmits $52$-Kbit data to distribute the precoding matrix for each UE. The whole precoding procedure, including uploading the CSI, calculating the precoding matrices at the CN, and distributing the precoding matrices, should be performed per $0.5$ ms. Instead, in our proposed decentralized scheme, only $130$-Kbit channel covariance data for a UE is transmitted in the backhaul by utilizing the conjugate property of channel covariance matrices, and the fronthaul also uses $130$-Kbit data to send the channel covariance matrix from each UE in other BSs. Moreover, the channel covariance matrices remain unchanged, and thus in principle no update is needed. In practical systems, the covariance matrices can be updated to cope with various system dynamics, e.g., user mobility and scheduling, according to the sounding period, i.e., each BS communicates with the CN only per $10 \sim 160$ ms.
\begin{figure}[t]
\centering
\includegraphics[width=3.4in]{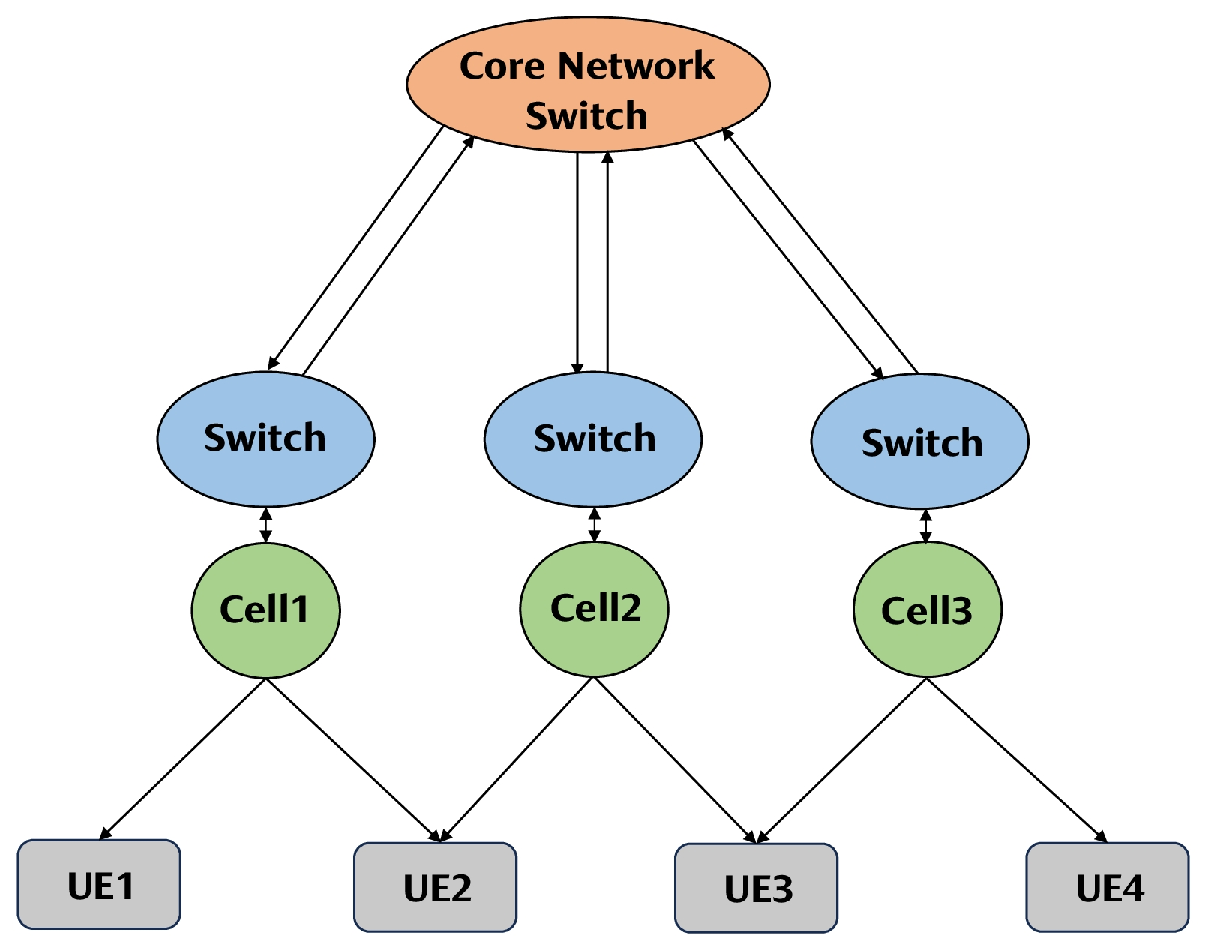}
\caption{The IP-RAN architecture for the proposed decentralized CJT precoding scheme.}
\label{IPRAN}
\end{figure}

\section{Fast ADMM-based CCCP Precoding Solver Design \label{sectioniv}}
Since the inter-cell interference bounds can be determined via DE, the coordinated beamforming problem (\ref{eq:opt2}) can be decoupled to (\ref{eq:opt3}) and solved with many approaches. Although Problem (\ref{eq:opt3}) can be transformed into an SOCP problem in (\ref{SOCP}), where the optimal solution can be found with interior point method, the computational complexity is too high for real-time execution. Therefore, considering the practical implementation restrictions, a new approach that can solve Problem (\ref{eq:opt3}) at each BS effectively and efficiently should be developed.

\subsection{Convex-concave Procedure (CCCP) for Problem Convexification}
We first rewrite the SINR constraints (\ref{eq:opt32}) as follows:
\begin{align} 
\label{DCP}
\begin{aligned}
    \gamma_{ip}\left(\sum_{j \in \mathcal{U}_p \setminus i} |\mathbf{h}_{ip}^H \mathbf{w}_{jp}|^2 + \sum_{q \in \mathcal{T}_i \setminus p} \tau_{iq} + \sum_{q \in \mathcal{T} \setminus  \mathcal{T}_i} \epsilon_{iq} + \sigma^2\right)\\
    - \left|\mathbf{h}_{ip}^H \mathbf{w}_{ip}\right|^2 \leq 0, \forall i \in \mathcal{U}_p,
\end{aligned}
\end{align}

\noindent which appears to be difference of convex (DC) \cite{ome2015} constraints. Then, by replacing the concave parts in the constraint (\ref{DCP}), i.e., $-|\mathbf{h}_{ip}^H \mathbf{w}_{ip}|^2$, with their first-order Taylor expansions, the CCCP convexifies the constraint and tackles Problem (\ref{eq:opt3}) by successively solving a series of convex subproblems \cite{mtao2016}, where a stationary point of Problem (\ref{eq:opt3}) is ensured \cite{gr2009}. In particular, in the $l$-th iteration, the following problem needs to be solved:
\begin{subequations} \label{eq:opt4*}
\begin{align}  
      \min_{\{\mathbf{w}_{ip} \}}  
    &\sum_{i \in \mathcal{U}_p} \Vert \mathbf{w}_{ip} \Vert_2^2 \label{eq:opt41} \\
      \text{s.t.}\  &\gamma_{ip}\left(\sum_{j \in \mathcal{U}_p \setminus i} |\mathbf{h}_{ip}^H \mathbf{w}_{jp}|^2 + \sum_{q \in \mathcal{T}_i \setminus p} \tau_{iq} + \sum_{q \in \mathcal{T} \setminus  \mathcal{T}_i} \epsilon_{iq} + \sigma^2\right) \notag\\
      &\ -2\Re{\{(\mathbf{w}_{ip,c}^{(l)})^{H}\mathbf{h}_{ip}\mathbf{h}_{ip}^{H}\mathbf{w}_{ip}\}}+|\mathbf{h}_{ip}^{H} \mathbf{w}_{ip}^{(l)}|^2 \leq 0,\ \forall i \in \mathcal{U}_p, \label{eq:opt42}\\
      &\ \sum_{j \in \mathcal{U}_p \setminus i } |\mathbf{h}_{ip}^H \mathbf{w}_{jp}|^2 \leq \tau_{ip}, \ \forall i \in \mathcal{U}_p, \label{eq:opt43}\\
      &\ \sum_{j \in \mathcal{U}_p} |\mathbf{h}_{ip}^H \mathbf{w}_{jp}|^2 \leq \epsilon_{ip}, \ \forall i \in \mathcal{U} \setminus \mathcal{U}_{p}, \label{eq:opt44}
\end{align}
\end{subequations}
\noindent where $\mathbf{w}_{ip}^{(l)}$ denotes the optimal solution in the previous CCCP iteration. Since Problem (\ref{eq:opt4*}) is convex, we next propose an ADMM-based solver to find a solution efficiently.

\subsection{ADMM-based Solver Design}
In order to design the ADMM-based solver, we introduce two sets of auxiliary variables as follows:
\begin{align}
A_{i,j} = \mathbf{h}_{ip}^{H}\mathbf{w}_{jp}, \forall i,j \in \mathcal{U}_{p},
\end{align}
\begin{align}
B_{i,j} = \mathbf{h}_{ip}^{H}\mathbf{w}_{jp}, \forall i \in \mathcal{U} \setminus \mathcal{U}_{p}, j \in \mathcal{U}_{p},
\end{align}

\noindent where $A_{i,j}$ denotes the received interference or signal term at user $i$ caused by the precoding vector of user $j$ (both user $i$ and $j$ are served by BS $p$), and $B_{i,j}$ denotes the received interference term at user $i$ caused by the precoding vector of user $j$ (user $i$ is not served by BS $p$ while user $j$ is served by BS $p$). Thus, Problem (\ref{eq:opt4*}) can be reformulated as follows:
\begin{subequations} \label{eq:opt5}
\begin{align}  
      &\min_{\{A_{i,j}, B_{i,j},\mathbf{w}_{ip} \}}  
    \sum_{i \in \mathcal{U}_p} \Vert \mathbf{w}_{ip} \Vert_2^2 \label{eq:opt51} \\
      \text{s.t.}\ \ &A_{i,j} - \mathbf{h}_{ip}^{H}\mathbf{w}_{jp}=0, \forall i,j \in \mathcal{U}_{p}, \label{eq:opt52}\\
      &\ B_{i,j} - \mathbf{h}_{ip}^{H}\mathbf{w}_{jp}=0, \forall i \in \mathcal{U} \setminus \mathcal{U}_{p}, j \in \mathcal{U}_{p}, \label{eq:opt53}\\
      &\gamma_{ip}\left(\sum_{j \in \mathcal{U}_p \setminus i} |A_{i,j}|^2 + \sum_{q \in \mathcal{T}_i \setminus p} \tau_{iq} + \sum_{q \in \mathcal{T} \setminus  \mathcal{T}_i} \epsilon_{iq} + \sigma^2\right) \notag\\
      &-2\Re{\{(\mathbf{w}_{ip,c}^{(l)})^{H}\mathbf{h}_{ip}A_{i,i}\}}+|\mathbf{h}_{ip}^{H} \mathbf{w}_{ip}^{(l)}|^2 \leq 0,\ \forall i \in \mathcal{U}_p, \label{eq:opt54}\\
      &\sum_{j \in \mathcal{U}_p \setminus i } |A_{i,j}|^2 \leq \tau_{ip}, \ \forall i \in \mathcal{U}_p, \label{eq:opt55}\\
      &\sum_{j \in \mathcal{U}_p} |B_{i,j}|^2 \leq \epsilon_{ip}, \ \forall i \in \mathcal{U} \setminus \mathcal{U}_{p}. \label{eq:opt56}
\end{align}
\end{subequations}

We define the feasible regions of constraint (\ref{eq:opt54}) and (\ref{eq:opt55}) as $\mathcal{E}$ and $\mathcal{F}$, respectively, and their intersection is denoted as $\mathcal{G} \triangleq \mathcal{E} \cap \mathcal{F}$. The indicator function of $\mathcal{G}$ is defined as follows:
\begin{align}
\mathbb{I}_{\mathcal{G}}(\{A_{i,j}\})= \begin{cases}0, & \text { if } \{A_{i,j}\}  \in \mathcal{G}, \\ +\infty, & \text { otherwise. }\end{cases}
\end{align}
\noindent Similarly, the feasible region of constraint (\ref{eq:opt56}) is denoted as $\mathcal{H}$ with its indicator function expressed as follows:
\begin{align}
\mathbb{I}_{\mathcal{H}}(\{B_{i,j}\})= \begin{cases}0, & \text { if } \{B_{i,j}\}  \in \mathcal{H}, \\ +\infty, & \text { otherwise. }\end{cases}
\end{align}
\noindent Therefore, Problem (\ref{eq:opt5}) can be rewritten as follows:
\begin{subequations} \label{eq:opt6}
\begin{align}  
      \min_{\{A_{i,j}, B_{i,j},\mathbf{w}_{ip} \}}  
    &\sum_{i \in \mathcal{U}_p} \Vert \mathbf{w}_{ip} \Vert_2^2 + \mathbb{I}_{\mathcal{G}}(\{A_{i,j}\}) + \mathbb{I}_{\mathcal{H}}(\{B_{i,j}\})  \label{eq:opt61} \\
      \text{s.t.}\ \ \ &A_{i,j} - \mathbf{h}_{ip}^{H}\mathbf{w}_{jp}=0, \forall i,j \in \mathcal{U}_{p}, \label{eq:opt62}\\
      &\ B_{i,j} - \mathbf{h}_{ip}^{H}\mathbf{w}_{jp}=0, \forall i \in \mathcal{U} \setminus \mathcal{U}_{p}, j \in \mathcal{U}_{p}, \label{eq:opt63}
\end{align}
\end{subequations}
\noindent and its augmented Lagrangian function is expressed as follows:
\begin{align}
\begin{aligned}
\label{augmentLA}
&\quad L(\{A_{i,j}\},\{B_{i,j}\},\{\mathbf{w}_{ip}\},\{\hat{\lambda}_{i,j}\},\{\hat{\mu}_{i,j}\}) \\
&= \sum_{i \in \mathcal{U}_p} \Vert \mathbf{w}_{ip} \Vert_2^2 + \mathbb{I}_{\mathcal{G}}(\{A_{i,j}\}) + \mathbb{I}_{\mathcal{H}}(\{B_{i,j}\}) \\
&+ \frac{\rho_1}{2}\sum_{i\in \mathcal{U}_{p}}\sum_{j\in \mathcal{U}_{p}}|A_{i,j}-\mathbf{h}_{ip}^{H}\mathbf{w}_{jp}+\hat{\lambda}_{i,j}|^{2}\\
&+ \frac{\rho_2}{2}\sum_{i \in \mathcal{U} \setminus \mathcal{U}_{p}}\sum_{j\in \mathcal{U}_{p}}|B_{i,j}-\mathbf{h}_{ip}^{H}\mathbf{w}_{jp}+\hat{\mu}_{i,j}|^{2},
\end{aligned}
\end{align}
\noindent where $\rho_1 >0$ and $\rho_2 >0$ are the penalty parameters. Besides, $\{\hat{\lambda}_{i,j}\}_{i,j\in \mathcal{U}_{p}}$ and $\{\hat{\mu}_{i,j}\}_{i \in \mathcal{U} \setminus \mathcal{U}_{p},j\in \mathcal{U}_{p}}$ denote the scaled dual variables associated with constraints (\ref{eq:opt62}) and (\ref{eq:opt63}), respectively.

The three sets of variables $\{A_{i,j}\}$, $\{B_{i,j}\}$, and $\{\mathbf{w}_{ip}\}$ can be updated alternatively according to the augmented Lagrangian function (\ref{augmentLA}). Details of updating the three sets of variables are illustrated in the following part, where $m$ denotes the ADMM iteration number.

1) Update $\{A_{i,j}\}$: The update of $\{A_{i,j}\}$ should follow
\begin{align}
\begin{aligned}
\{A_{i,j}^{(l,m+1)}\}_{i,j \in \mathcal{U}_{p}} &= \mathop{\arg \min}_{\{A_{i,j}\}}\Bigg\{\mathbb{I}_{\mathcal{G}}(\{A_{i,j}\})\\
&+\frac{\rho_1}{2}\sum_{i\in \mathcal{U}_{p}}\sum_{j\in \mathcal{U}_{p}}|A_{i,j}\!-\!\mathbf{h}_{ip}^{H}\mathbf{w}_{jp}^{(l,m)}\!+\!\lambda_{i,j}^{(l,m)}|^{2}\Bigg\},
\end{aligned}
\end{align}
\noindent which can be decomposed to $|\mathcal{U}_{p}|$ sub-problems. For each $i\in \mathcal{U}_{p}$, the sub-problem can be transformed to the following problem:
\begin{subequations} \label{eq:opt7}
\begin{align}  
      \min_{\{A_{i,j}\}_{j=1}^{|\mathcal{U}_{p}|}}  
    &\sum_{j\in \mathcal{U}_{p}}|A_{i,j}-\mathbf{h}_{ip}^{H}\mathbf{w}_{jp}^{(l,m)}+\lambda_{i,j}^{(l,m)}|^{2}  \label{eq:opt71} \\
      \text{s.t.}\ \ \ &\ \gamma_{ip}\left(\sum_{j \in \mathcal{U}_p \setminus i} |A_{i,j}|^2 + \sum_{q \in \mathcal{T}_i \setminus p} \tau_{iq} + \sum_{q \in \mathcal{T} \setminus  \mathcal{T}_i} \epsilon_{iq} + \sigma^2\right) \notag\\
      &\ -2\Re{\{(\mathbf{w}_{ip,c}^{(l)})^{H}\mathbf{h}_{ip}A_{i,i}\}}+|\mathbf{h}_{ip}^{H} \mathbf{w}_{ip}^{(l)}|^2 \leq 0, \label{eq:opt72} \\
      &\ \sum_{j \in \mathcal{U}_p \setminus i } |A_{i,j}|^2 \leq \tau_{ip}. \label{eq:opt73}
\end{align}
\end{subequations}

For this problem, the closed-form optimal solution for each $i \in \mathcal{U}_{p}$ is given as follows:
\begin{align}
A_{i,j}^{(l,m+1)}= \begin{cases}\frac{\sqrt{\tau_{ip}}\left(\mathbf{h}_{ip}^{H}\mathbf{w}_{jp}^{(l,m)}-\lambda_{i,j}^{(l,m)}\right)}{\sqrt{\sum_{j \in \mathcal{U}_p \setminus i}|\mathbf{h}_{ip}^{H}\mathbf{w}_{jp}^{(l,m)}-\lambda_{i,j}^{(l,m)}|^{2}}}, & \text { if } j \neq i, \\ \zeta_{ip}^{(l,m)}\mathbf{h}_{ip}^{H}\mathbf{w}_{jp,c}^{(l)} +\mathbf{h}_{i}^{H}\mathbf{w}_{jp}^{(l,m)}-\lambda_{i,j}^{(l,m)}, & \text { if } j = i,\end{cases}
\end{align}
\noindent where $\zeta_{i}^{(l,m)}=\Big(\gamma_{ip}(\tau_{ip} + \sum_{q \in \mathcal{T}_i \setminus p} \tau_{iq} + \sum_{q \in \mathcal{T} \setminus  \mathcal{T}_i} \epsilon_{iq} + \sigma^2)+|\mathbf{h}_{ip}^{H}\mathbf{w}_{ip,c}^{(l)}|^2-2\Re\{(\mathbf{w}_{ip,c}^{(l)})^{H}\mathbf{h}_{ip}(\mathbf{h}_{ip}^{H}\mathbf{w}_{ip}^{(l,m)}-\lambda_{i,i}^{(l,m)})\}\Big)\Big/\Big(2|\mathbf{h}_{ip}^{H}\mathbf{w}_{ip,c}^{(l)}|^2\Big)$. The detailed derivations are given in Appendix \ref{appendixupdateA}.

2) Update $\{B_{i,j}\}$: Variables $\{B_{i,j}\}$ are updated by solving the following problem:
\begin{align}
\begin{aligned}
&\{B_{i,j}^{(l,m+1)}\}_{i \in \mathcal{U} \setminus \mathcal{U}_{p}, j \in \mathcal{U}_{p}} = \mathop{\arg \min}_{\{B_{i,j}\}}\Bigg\{\mathbb{I}_{\mathcal{H}}(\{B_{i,j}\})\\
&+\frac{\rho_2}{2}\sum_{i\in \mathcal{U} \setminus \mathcal{U}_{p}}\sum_{j\in \mathcal{U}_{p}}|B_{i,j}-\mathbf{h}_{ip}^{H}\mathbf{w}_{jp}^{(l,m)}+\mu_{i,j}^{(l,m)}|^{2}\Bigg\},
\end{aligned}
\end{align}
\noindent which can be decomposed to $|\mathcal{U}|-|\mathcal{U}_{p}|$ sub-problems. For each $i\in \mathcal{U} \setminus \mathcal{U}_{p}$, the sub-problem can be transformed to the following problem:
\begin{subequations} \label{eq:opt8}
\begin{align}  
      \min_{\{B_{i,j}\}_{j=1}^{|\mathcal{U}_{p}|}}  
    &\sum_{j\in \mathcal{U}_{p}}|B_{i,j}-\mathbf{h}_{ip}^{H}\mathbf{w}_{jp}^{(l,m)}+\mu_{i,j}^{(l,m)}|^{2}  \label{eq:opt81} \\
      \text{s.t.}\ \ \ &\sum_{j \in \mathcal{U}_p} |B_{i,j}|^2 \leq \epsilon_{ip}, \label{eq:opt82}
\end{align}
\end{subequations}
\noindent which can also be solved optimally with closed-form solutions. Specifically, we define $\mathbf{B}_{i}=[B_{i,1},\cdots,B_{i,|\mathcal{U}_{p}|}]^{T}$ and $\mathbf{C}_{i}^{(l,m)}=[\mathbf{h}_{ip}^{H}\mathbf{w}_{1p}^{(l,m)}-\mu_{i,1}^{(l,m)},\cdots,\mathbf{h}_{ip}^{H}\mathbf{w}_{|\mathcal{U}_{p}|p}^{(l,m)}-\mu_{i,|\mathcal{U}_{p}|}^{(l,m)}]^{T}$. Then, Problem (\ref{eq:opt8}) is reformulated as follows:
\begin{subequations} \label{eq:opt9}
\begin{align}  
      \min_{\mathbf{B}_{i}}\  
    &||\mathbf{B}_{i}-\mathbf{C}_{i}^{(l,m)}||_2  \label{eq:opt91} \\
      \text{s.t.}\ &||\mathbf{B}_{i}||_2 \leq \sqrt{\epsilon_{ip}}, \label{eq:opt92}
\end{align}
\end{subequations}
\noindent which can be viewed as the Euclidean projection of the point $\mathbf{C}_{i}^{(l,m)}$ onto an Euclidean ball with radius $\sqrt{\epsilon_{ip}}$. Thus, the optimal solution can be simply given by the projection of $\mathbf{C}_{i}^{(l,m)}$ on an Euclidean ball, which takes the origin of coordinate as the center and has the radius of $\sqrt{\epsilon_{ip}}$, as follows:
\begin{align}
\mathbf{B}_{i}^{(l,m+1)}= \mathop{\min}\left\{\frac{\sqrt{\epsilon_{ip}}}{||\mathbf{C}_{i}^{(l,m)}||_2},1\right\} \mathbf{C}_{i}^{(l,m)}.
\end{align}

3) Update $\{\mathbf{w}_{ip}\}$: For each $i \in \mathcal{U}_{p}$, the update of $\mathbf{w}_{ip}$ is based on the following formula:
\begin{align}
\begin{aligned}
&\mathbf{w}_{ip}^{(l,m+1)} =\\
& \mathop{\arg \min}_{\mathbf{w}_{ip}}\Big\{\Vert \mathbf{w}_{ip} \Vert_2^2 + \frac{\rho_1}{2}\sum_{j\in \mathcal{U}_{p}}|A_{j,i}^{(l,m+1)}-\mathbf{h}_{jp}^{H}\mathbf{w}_{ip}+\lambda_{j,i}^{(l,m)}|^{2}\\ 
&+ \frac{\rho_2}{2}\sum_{j\in \mathcal{U} \setminus \mathcal{U}_{p}}|B_{j,i}^{(l,m+1)}-\mathbf{h}_{jp}^{H}\mathbf{w}_{ip}+\mu_{j,i}^{(l,m)}|^{2}\Big\}.
\end{aligned}
\end{align}

\noindent This problem can be solved by finding the stationary point of the objective function, which is expressed as follows:
\begin{align}
\begin{aligned}
\mathbf{w}_{ip}^{(l,m+1)}&=\left(2\mathbf{I}_{N_T}+\rho_1 \sum_{j\in \mathcal{U}_{p}} \mathbf{h}_{jp}\mathbf{h}_{jp}^{H}+ \rho_2 \sum_{j \in \mathcal{U} \setminus \mathcal{U}_{p}} \mathbf{h}_{jp}\mathbf{h}_{jp}^{H}\right)^{-1} \\
&\times \Bigg(\rho_1 \sum_{j\in \mathcal{U}_{p}} \mathbf{h}_{jp}(A_{j,i}^{(l,m+1)}+\lambda_{j,i}^{(l,m)})\\
&+\rho_2 \sum_{j \in \mathcal{U} \setminus \mathcal{U}_{p}} \mathbf{h}_{jp}(B_{j,i}^{(l,m+1)}+\mu_{j,i}^{(l,m)})\Bigg).
\end{aligned}
\end{align}
\noindent We note that $\left(2\mathbf{I}_{N_T}\!+\!\rho_1\!\sum_{j\in \mathcal{U}_{p}}\!\mathbf{h}_{jp}\mathbf{h}_{jp}^{H}\!+\!\rho_2\!\sum_{j \in \mathcal{U} \setminus \mathcal{U}_{p}}\!\mathbf{h}_{jp}\mathbf{h}_{jp}^{H}\!\right)^{-1}$ only needs to be computed once, which does not bring additional computational burden to the ADMM iterations.

Finally, the last step is to update the dual variables as follows:
\begin{align}
\begin{aligned}
&\lambda_{i,j}^{(l,m+1)}=\lambda_{i,j}^{(l,m)} + A_{i,j}^{(l,m+1)} - \mathbf{h}_{ip}^{H}\mathbf{w}_{jp}^{(l,m+1)},\ \forall i, j \in \mathcal{U}_{p},\\
&\mu_{i,j}^{(l,m+1)}=\mu_{i,j}^{(l,m)} + B_{i,j}^{(l,m+1)}- \mathbf{h}_{ip}^{H}\mathbf{w}_{jp}^{(l,m+1)}, \\
&\quad \quad \quad \quad \quad \quad \quad \quad \quad \quad \quad \quad \quad \quad \quad \forall i \in \mathcal{U} \setminus \mathcal{U}_{p}, j \in \mathcal{U}_{p}.
\end{aligned}
\end{align}

The workflow of the CCCP-ADMM solver for decentralized user-centric CJT precoding of BS $p$ is summarized in Algorithm \ref{algo-cccpadmm}. We use the local ZF method with only local CSI to initialize $\{\mathbf{w}_{ip,c}^{(1)}\}$ of Algorithm \ref{algo-cccpadmm}. When the ADMM algorithm in each CCCP iteration starts, the solution obtained in the previous CCCP iteration is used to initialize $\{\mathbf{w}_{ip}^{(l,1)}\}$.
\begin{algorithm}[htpb]
\caption{CCCP-ADMM Solver for Decentralized CJT Precoding of BS $p$ \label{algo-cccpadmm}}
{\bf Input:}
Local CSI of BS $p$ $\mathbf{h}_{ip}$, $\forall i \in \mathcal{U}_{p}$, target SINR $\gamma_{ip}$, $\forall i \in \mathcal{U}_q$, interference bounds for BS $p$ $\tau_{ip}$, $\forall i \in \mathcal{U}_{p}$, and $\epsilon_{ip}$, $\forall i \in \mathcal{U} \setminus \mathcal{U}_{p}$, interference bounds from other BSs $\tau_{iq}$, $\forall i \in \mathcal{U}_{p}$, $\forall q \in \mathcal{T}_{i} \setminus p$, and $\epsilon_{iq}$, $\forall i \in \mathcal{U}_{p}$, $\forall q \in \mathcal{T} \setminus \mathcal{T}_{i}$, AWGN noise variance $\sigma^2_{i}$, $\forall i \in \mathcal{U}_{p}$, regularization parameters $\rho_1$ and $\rho_2$, maximum number of CCCP iterations $Q_{1}$, maximum number of ADMM iterations $Q_{2}$, and accuracy tolerance $\epsilon$.\\
{\bf Output:}
Local precoding vectors $\mathbf{w}_{ip}=\mathbf{w}_{ip,c}^{(l+1)}$, $\forall i \in \mathcal{U}_{p}$.\\
{\bf Initialize:}
$l \leftarrow 1$, $\mathbf{w}_{ip,c}^{(1)}=\mathbf{w}_{ip}^{ZF}$, $\forall i \in \mathcal{U}_{p}$.
\begin{algorithmic}[1]
\STATE $\mathbf{R}_{p}=\left(2\mathbf{I}_{N_T}\!+\!\rho_1 \sum_{j\in \mathcal{U}_{p}} \mathbf{h}_{jp}\mathbf{h}_{jp}^{H}\!+\!\rho_2 \sum_{j \in \mathcal{U} \setminus \mathcal{U}_{p}} \mathbf{h}_{jp}\mathbf{h}_{jp}^{H}\right)^{-1}$
\WHILE{{$l \leq Q_{1}$} \text{and} {$\frac{\sum_{i}||\mathbf{w}_{ip,c}^{(l)}-\mathbf{w}_{ip,c}^{(l-1)}||_{2}^{2}}{\sum_{i}||\mathbf{w}_{ip,c}^{(l-1)}||_{2}^{2}} > \epsilon$}}
\STATE {\bf Initialize:} $m \leftarrow 1$, $\lambda_{i,j}^{(l,1)}=0$, $\forall i, j \in \mathcal{U}_{p}$, $\mu_{i,j}^{(l,1)}=0$, 
\Statex \quad $\forall i \in \mathcal{U} \setminus \mathcal{U}_{p}, j \in \mathcal{U}_{p}$, $\mathbf{w}_{ip}^{(l,1)}=\mathbf{w}_{ip,c}^{(l)}$, $\forall i \in \mathcal{U}_{p}$.
\WHILE{{$m \leq Q_{2}$}}
\STATE $\zeta_{i}^{(l,m)}=\Big(\gamma_{ip}(\tau_{ip} + \sum_{q \in \mathcal{T}_i \setminus p} \tau_{iq} + \sum_{q \in \mathcal{T} \setminus  \mathcal{T}_i} \epsilon_{iq}+ \sigma^2) $
\Statex \quad \quad $+|\mathbf{h}_{ip}^{H}\mathbf{w}_{ip,c}^{(l)}|^2-2\Re\{(\mathbf{w}_{ip,c}^{(l)})^{H}\mathbf{h}_{ip}(\mathbf{h}_{ip}^{H}\mathbf{w}_{ip}^{(l,m)}$
\Statex \quad \quad $-\lambda_{i,i}^{(l,m)})\}\Big)\Big/\Big(2|\mathbf{h}_{ip}^{H}\mathbf{w}_{ip,c}^{(l)}|^2\Big)$, $\forall i \in \mathcal{U}_{p}$
\STATE $A_{i,j}^{(l,m+1)}\!=\!\begin{cases}\frac{\sqrt{\tau_{ip}}\left(\mathbf{h}_{ip}^{H}\mathbf{w}_{jp}^{(l,m)}-\lambda_{i,j}^{(l,m)}\right)}{\sqrt{\sum_{j \in \mathcal{U}_p \setminus i}|\mathbf{h}_{ip}^{H}\mathbf{w}_{jp}^{(l,m)}-\lambda_{i,j}^{(l,m)}|^{2}}}, \text { if }\!j\!\neq\!i, \\ \!\zeta_{ip}^{(l,m)}\mathbf{h}_{ip}^{H}\mathbf{w}_{jp,c}^{(l)} \!+\!\mathbf{h}_{i}^{H}\mathbf{w}_{jp}^{(l,m)}\!-\!\lambda_{i,j}^{(l,m)}\!,\!\text { if }\!j\!=\!i,\end{cases}$ 
\Statex \quad \quad \quad \quad \quad \quad \quad \quad \quad \quad \quad \quad \quad \quad \quad \quad \quad \quad \quad $\forall i,j \in \mathcal{U}_{p}$
\STATE $\mathbf{B}_{i}^{(l,m+1)}= \mathop{\min}\left\{\frac{\sqrt{\epsilon_{ip}}}{||\mathbf{C}_{i}^{(l,m)}||_2}, 1\right\} \mathbf{C}_{i}^{(l,m)}$, $\forall i \in \mathcal{U}_{p}$, with 
\Statex \quad \quad $\mathbf{B}_{i}=[B_{i,1},\cdots,B_{i,|\mathcal{U}_{p}|}]^{T}$ and $\mathbf{C}_{i}^{(l,m)}=[\mathbf{h}_{ip}^{H}\mathbf{w}_{1p}^{(l,m)}$
\Statex \quad \quad $-\mu_{i,1}^{(l,m)},\cdots,\mathbf{h}_{ip}^{H}\mathbf{w}_{|\mathcal{U}_{p}|p}^{(l,m)}-\mu_{i,|\mathcal{U}_{p}|}^{(l,m)}]^{T}$
\STATE $\mathbf{w}_{ip}^{(l,m+1)} = \mathbf{R}_{p} \Bigg(\rho_1 \sum_{j\in \mathcal{U}_{p}} \mathbf{h}_{jp}(A_{j,i}^{(l,m+1)}+\lambda_{j,i}^{(l,m)})$
\Statex \quad \quad $+\rho_2 \sum_{j \in \mathcal{U} \setminus \mathcal{U}_{p}} \mathbf{h}_{jp}(B_{j,i}^{(l,m+1)}+\mu_{j,i}^{(l,m)})\Bigg)$, $\forall i \in \mathcal{U}_{p}$
\STATE $\lambda_{i,j}^{(l,m+1)}=\lambda_{i,j}^{(l,m)} + A_{i,j}^{(l,m+1)} - \mathbf{h}_{ip}^{H}\mathbf{w}_{jp}^{(l,m+1)},\ \forall i, j \in \mathcal{U}_{p}$
\STATE $\mu_{i,j}^{(l,m+1)}=\mu_{i,j}^{(l,m)} + B_{i,j}^{(l,m+1)}- \mathbf{h}_{ip}^{H}\mathbf{w}_{jp}^{(l,m+1)}$,
\Statex \quad \quad \quad \quad \quad \quad \quad \quad \quad \quad \quad \quad \quad \quad \quad $\forall i \in \mathcal{U} \setminus \mathcal{U}_{p}, j \in \mathcal{U}_{p}$
\STATE $m \leftarrow m+1$
\ENDWHILE
\STATE $\mathbf{w}_{ip,c}^{(l+1)}=\mathbf{w}_{ip}^{(l,m)}$, $\forall i \in \mathcal{U}_{p}$
\STATE $l \leftarrow l+1$
\ENDWHILE
\end{algorithmic}
\end{algorithm}

\begin{remark}
The ADMM-based solver can converge to the global optimum of Problem (\ref{eq:opt5}) \cite{dp1989}. Besides, the computational complexity of one ADMM iteration is given as $\mathcal{O}(N_{T}^3)$. Moreover, in simulations, we find that by adapting the regularization parameters $\rho_1$ and $\rho_2$, our proposed ADMM solver is very efficient and no iteration is required (i.e., $Q_1 = Q_2 =1$) to achieve the near-optimal performance. Instead, the off-the-shelf SOCP solver in the CVX toolbox \cite{cvx} applies the interior point method with many iterations.
\end{remark}

\section{Simulation Results \label{sectionv}}
We simulate a multi-cell downlink celluar network with $N_B=3$ BSs and $N_c=20$ UEs, and the UEs are uniformly distributed in the coverage area of BSs. In particular, BS 1, BS 2, and BS 3 serve UEs 1-10, UEs 6-15, and UEs 11-20, respectively, which means that UEs 6-15 are served by two BSs simutaneously. The channel vectors $\{\mathbf{h}_{ip}\}$ are generated from the Quasi Deterministic Radio Channel Generator (QuaDRiGa) \cite{fbu2014}, which is calibrated against the 3rd Generation Partnership Project (3GPP) channel model. The noise variance is the same for all UEs, which is given by $\sigma^2=10^{\frac{1}{U_p} \sum_{p \in \mathcal{T}}\sum_{i \in \mathcal{U}_p} \log _{10}\left\|\mathbf{h}_{ip}\right\|_2^2} \times 10^{-\frac{\mathrm{SNR}}{10}}$ and $\mathrm{SNR} = 20$ dB is the average receive SNR for all users without precoding. Besides, $Q_1 = Q_2 =1$, $\rho_1=\rho_2=0.5$, and the accuracy
tolerance is $\epsilon=10^{-8}$. Our simulation results are averaged over 100 random channel realizations. Apart from our proposed low-complexity ADMM solver for the sub-problem at each BS, we also use the SOCP solver from the CVX toolbox for comparison. Besides, we simulate two baseline methods and a performance upper bound as follows:
\begin{itemize}
    \item \textbf{ZF-based decentralized precoding:} Each BS performs local ZF without utilizing the instantaneous CSI from other BSs.
\end{itemize}
\begin{itemize}
    \item \textbf{ZF-based centralized precoding:} A central unit collects the instantaneous CSI from all BSs and performs ZF. The precoding matrix is then distributed to each BS.
\end{itemize}
\begin{itemize}
    \item \textbf{UDD-based centralized precoding:} This scheme also requires the instantaneous CSI from all BSs and then generates the precoding matrices with the method presented in Section \ref{sec:central}. Furthermore, the target SINR constraints are generated from the WMMSE algorithm \cite{qshi2011}, i.e., using the WMMSE algorithm to generate the precoder and calculate the corresponding approximate SINR in (\ref{eq:SINR}). This can be viewed as a performance upper bound.
\end{itemize}

We evaluate two performance metrics namely the total transmit power of all BSs and the sum rate of all users. Note that the sum rate of all users are calculated based on the original SINR formulation in (\ref{eq:originalSINR}). When comparing the total transmit power, we ensure that the sum rate of all schemes are the same by normalizing the precoders to a specific total transmit power. Instead, when we are interested in the sum rate, a fixed total transmit power is required, e.g., all precoders are normalized to $10$ Watt in total power.

We first evaluate the total transmit power of all BSs versus the number of antennas at each BS in Fig. \ref{powerfig}. For convenience of display, the total transmit power of the ZF-based centralized precoding scheme is normalized to $10$ Watt for all sets of BS antenna numbers, which is the flat constant curve and used as a baseline such that the sum rate of all other precoding schemes can be tuned via total transmit power normalization to the same with the ZF-based centralized precoding scheme for each BS antenna number. It is seen that the ZF-based decentralized precoding scheme consumes much more power than other schemes to achieve a target sum rate, which demonstrates the performance gain brought by the coordination among BSs. In particular, when the number of BS antenna is smaller than that of the serving UE, i.e., $N_T=8$, the ZF-based decentralized precoding scheme cannot properly operate due to insufficient degree of freedom, and the performance gap becomes smaller with the number of BS antennas since more BS antennas enhance the performance of the ZF-based decentralized precoding scheme and diminish the contribution from coordination. Besides, compared with the ZF-based centralized precoding scheme, the proposed decentralized precoding scheme with the SOCP solver achieves much lower power consumption, and performs close to the lower bound achieved by the UDD-based precoder. This demonstrates the effectiveness of our proposed decentralized precoding scheme for CJT with only a little exchanged information by determining the interference bounds using the covariance matrices of global CSI via DE. Besides, the low-complexity ADMM solver without needing any iteration consumes at most $7\%$ more power compared with the SOCP solver, which proves its effectiveness in maintaining similar performance with the SOCP solver while greatly reducing computational complexity.
\begin{figure}[t]
\centering
\includegraphics[width=3.4in]{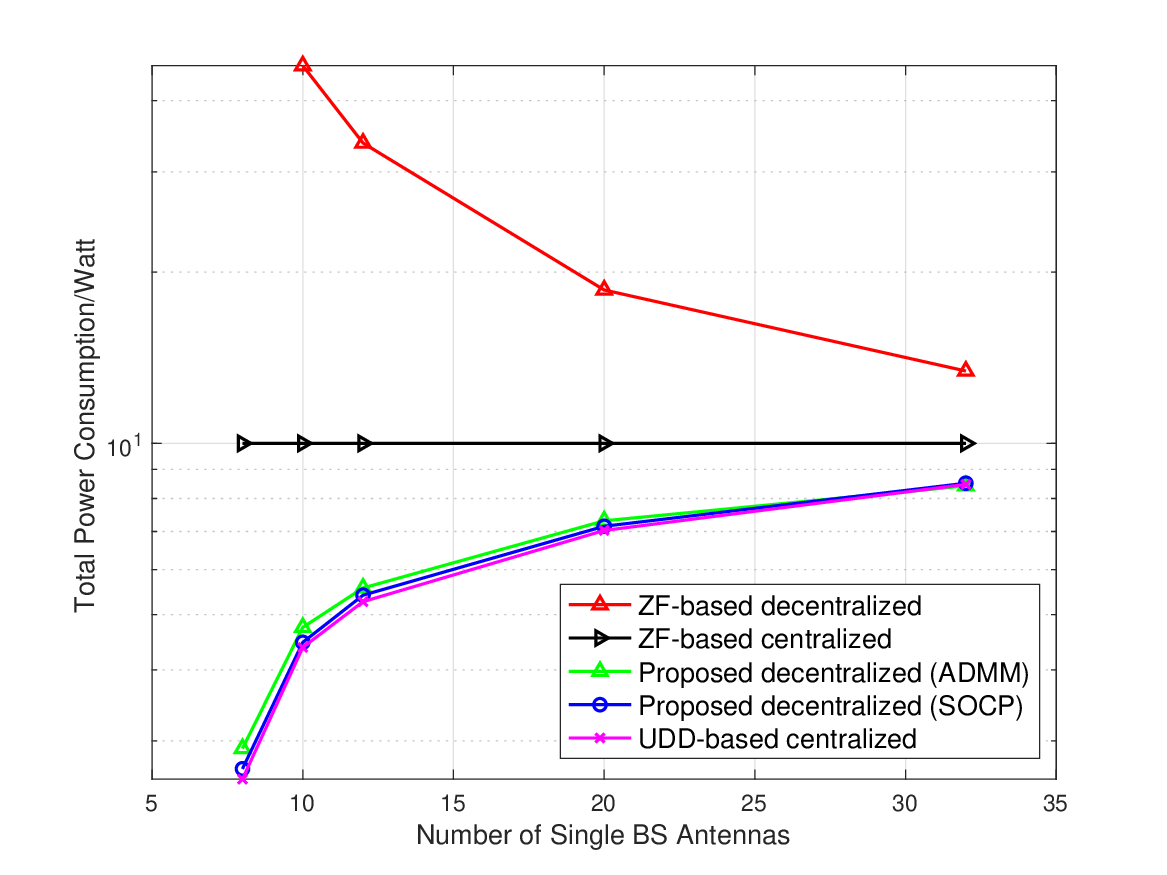}
\caption{Total power consumption versus the number of antennas at each BS.}
\label{powerfig}
\end{figure}

The sum rate of all users versus the number of antennas at each BS is shown in Fig. \ref{sumratefig}. Similar to the transmit power evaluation, the ZF-based decentralized precoding scheme presents the most inferior performance due to the lack of BS coordination. The proposed decentralized precoding schemes outperform the ZF-based centralized precoding scheme by $8\% \sim 57\%$ in terms of the sum rate. In particular, the decentralized precoding scheme with the SOCP solver achieves similar performance as that of the UDD-based optimal centralized precoding scheme. Besides, replacing the SOCP solver with the low-complexity ADMM solver only incurs no more than $5\%$. This again validates the effectiveness of the proposed decentralized precoding scheme by determining the interference bounds using the covariance matrices of global CSI via DE and the benefits of applying the ADMM solver.
\begin{figure}[t]
\centering
\includegraphics[width=3.4in]{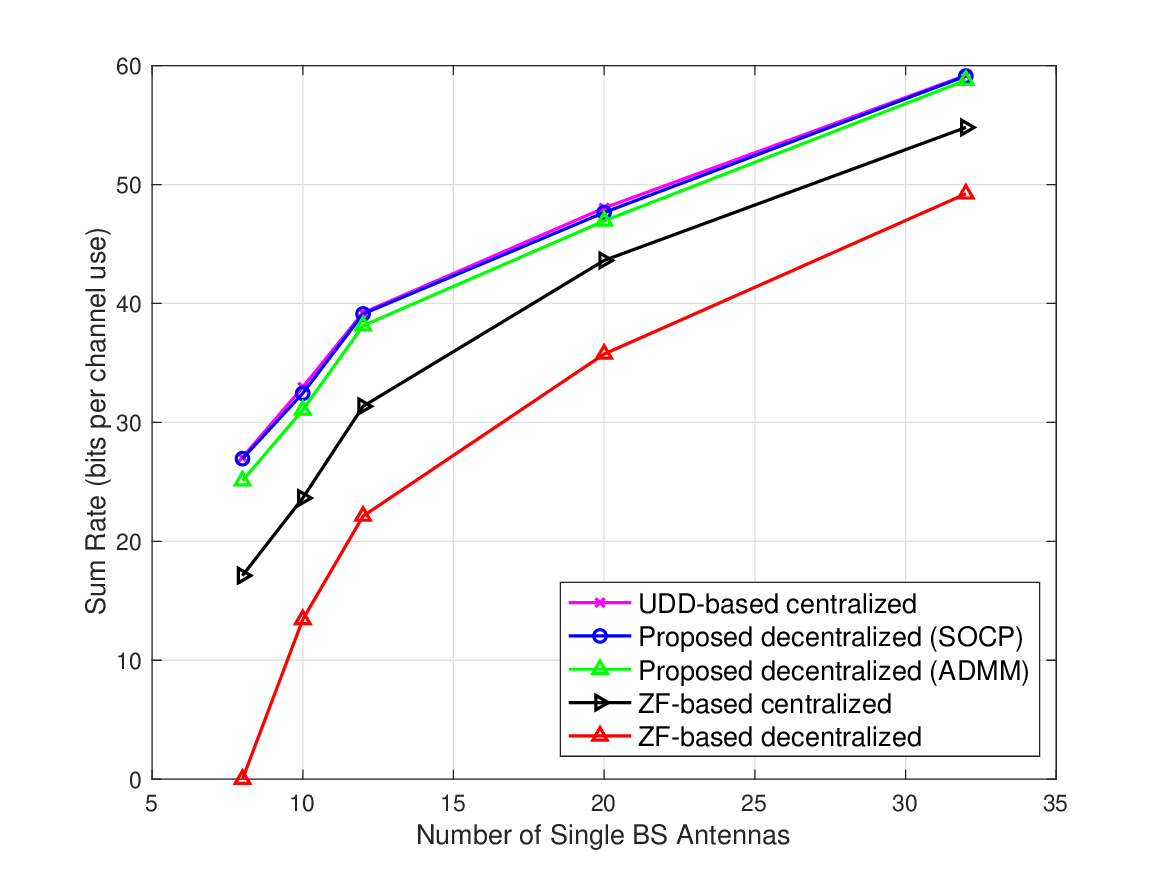}
\caption{Sum rate versus the number of antennas at each BS.}
\label{sumratefig}
\end{figure}

Since the proposed decentralized precoding scheme depends on the interference bounds and the target SINR constraints, we next evaluate the robustness of the proposed scheme with respect to the two parameters in Figs. \ref{robustinter} and  \ref{robustsinr}, where the number of BS antennas is $12$. In Fig. \ref{robustinter}, the interference bounds are scaled with a scalar $\alpha$ ranging from $0.2$ to $10$. It is observed that the sum rate of the proposed decentralized scheme with either the SOCP or ADMM solver decreases when the interference bounds are offset from the precise estimation obtained from DE. However, the proposed decentralized scheme still outperforms the ZF-based centralized precoding scheme. Besides, with the ADMM solver, the proposed decentralized precoding scheme achieves the near-optimal performance when the interference bounds are underestimated. However, adopting the SOCP solver appears to be a better option when interference bounds are overestimated. Similarly, the target SINR constraints are scaled with a factor $\beta$ ranging from $0.1$ to $10$, and the sum rate versus the scale factor $\beta$ for the target SINR constraints is shown in Fig. \ref{robustsinr}. It is seen that although there exists a performance degradation, the sum rate of the proposed decentralized scheme with either of the two solvers is larger than the ZF-based centralized precoding. In addition, the ADMM solver is much more robust than the SOCP solver when the target SINR is less than the optimal value. These results show that the proposed decentralized precoding scheme can maintain its superiority even when the interference bounds are wrongly estimated to some extent or the target SINR constraints from the WMMSE algorithm cannot be obtained, e.g., they may be obtained from the user scheduling information in practical systems, which demonstrates its robustness to the two input values, i.e., the interference bounds and the SINR constraints.
\begin{figure}[t]
\centering
\includegraphics[width=3.4in]{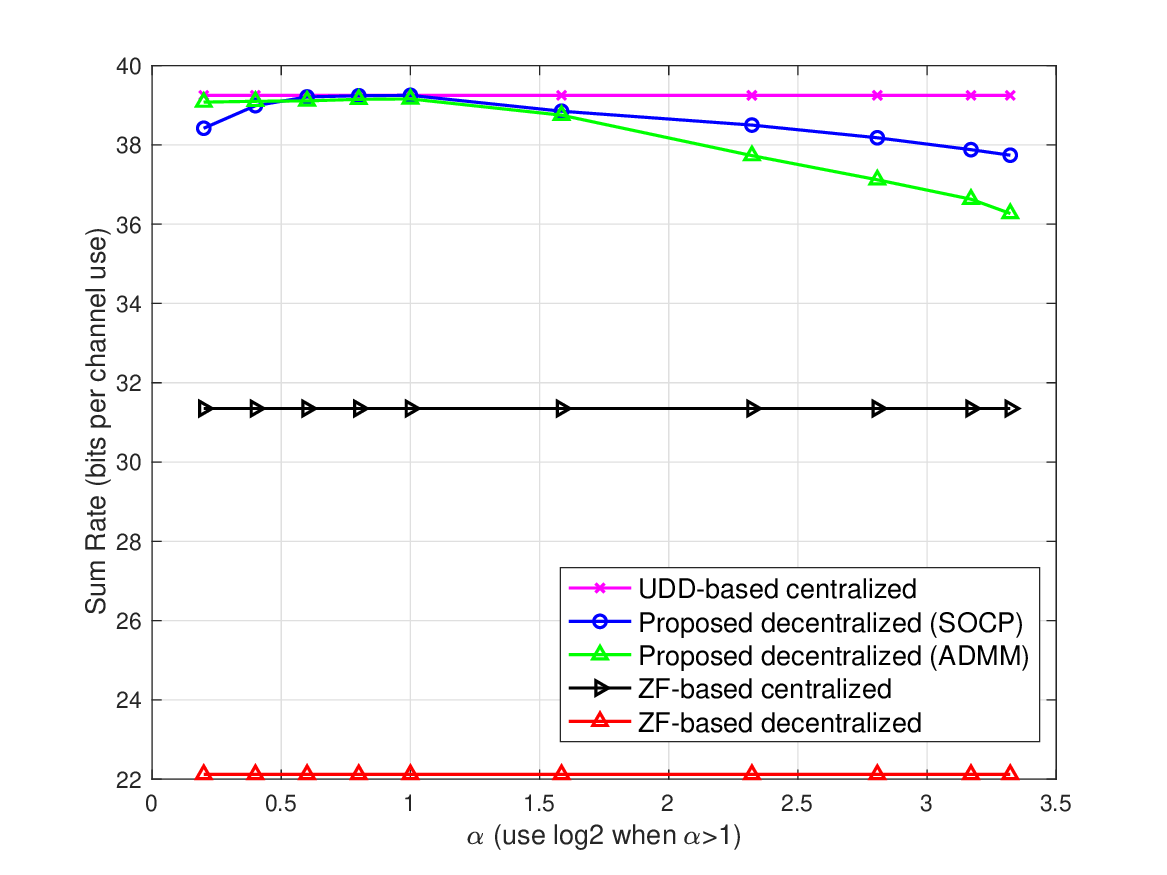}
\caption{Sum rate versus the scale factor of the interference bounds.}
\label{robustinter}
\end{figure}
\begin{figure}[t]
\centering
\includegraphics[width=3.4in]{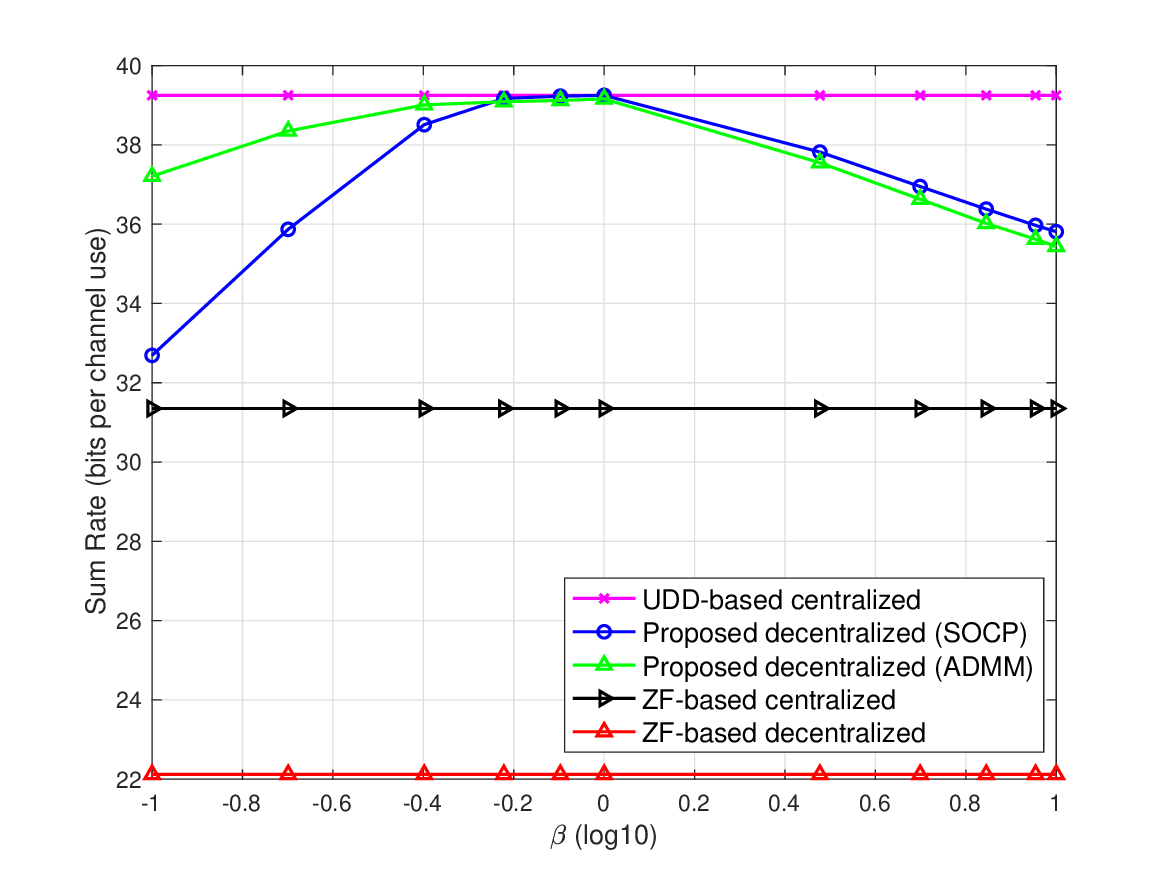}
\caption{Sum rate versus the scale factor of the target SINR.}
\label{robustsinr}
\end{figure}

We finally examine the computational complexity of all precoding schemes by comparing the average simulation running time on an HP Elitebook laptop (Model: 840 G6, i7-8665U Chip with 16GB RAM) in Table \ref{tabletime}. The number of BS antennas is also fixed as $12$. It is seen that the average running time of our proposed decentralized precoding scheme with the ADMM solver is only $1.7$ times of that of the ZF-based centralized precoding scheme, and is much less than that of the UDD-based precoder. However, the complexity of the ADMM solver is much lower in practice since the decentralized scheme allows all BSs to compute the precoding matrix of each BS in parallel, while the ZF-based centralized precoding can only calculate the precoding matrices of all BSs at a central unit. This is also valid for the ZF-based decentralized precoding scheme, which achieves the lowest complexity among all schemes. Besides, the SOCP solver bears a heavy computational complexity since it solves the problem with the complex interior point method, and the ADMM solver secures a $79\%$ complexity reduction compared with the SOCP solver. The above numerical results verify that the proposed decentralized precoding scheme with the ADMM solver can achieve near-optimal performance with extremely low complexity.
\begin{table}[t]
\caption{Average simulation running time of all precoding schemes}
\centering
\begin{tabular}{c|c}
    \hline
     Scheme & Average simulation running time (s)  \\
    \hline
    ZF-based decentralized & 0.47  \\
    \hline
    ZF-based centralized & 0.82 \\
    \hline
    UDD-based & 4.26 \\
    \hline
    Proposed (SOCP) & 6.88\\
    \hline
    Proposed (ADMM) & 1.41\\
    \hline
    \end{tabular}
\label{tabletime}
\end{table}

\section{Conclusions \label{sectionvi}}
In this paper, we proposed a decentralized precoding scheme for coherent joint transmission. In order to decouple the transmit power minimization problem at each base station, we first determined the bounds of the inter-cell interference by utilizing the deterministic equivalents from random matrix theory, which is solely based on statistical channel information, i.e., the covariance matrices of global channel state information. Subsequently, considering the heavy complexity of the off-the-shelf solver, we further developed a fast ADMM-based solver for each BS-based sub-problem, where the convex-concave procedure is utilized. Simulation results verified the effectiveness of our proposed decentralized precoding scheme and the corresponding solver, which achieves the near-optimal performance with low computational complexity. Future works can be focused on extending the proposed decentralized precoding scheme to the multi-stream scenario. In addition, the coordination base station selection problem can also be investigated to provide a guideline for the practical deployment of coherent joint transmission.

\appendices
\section{Lagrangian dual analysis for Problem (\ref{eq:opt2})} \label{sec:dual}
We first show the equivalence of Problem (\ref{eq:opt1}) and Problem (\ref{eq:opt2}) by proving that the equalities in (\ref{eq:opt23}) and (\ref{eq:opt24}) hold at the optimal solution of Problem (\ref{eq:opt2}). To do so, we write the Lagrangian function of (\ref{eq:opt2}) as follows:
\begin{equation} \label{eq:L1}
    \begin{aligned}
     &L\left(\{\textbf{w}_{ip}\}, \{\lambda_{ip}\}, \{\alpha_{iq}\}, \{\beta_{iq}\}, \{\epsilon_{iq}\}, \{\tau_{iq}\}\right) \\
     &= \sum_{p \in \mathcal{T}} \sum_{i \in \mathcal{U}_p} \Vert \textbf{w}_{ip}\Vert^2 + \sum_{q \in \mathcal{T}} \sum_{i \in \mathcal{U}_q} \alpha_{iq} \left(\sum_{j \in \mathcal{U}_q \setminus i } |\mathbf{h}_{iq}^H \mathbf{w}_{jq}|^2 - \tau_{iq}\right) \\
     &\!+\!\sum_{q \in \mathcal{T}} \sum_{i \notin \mathcal{U}_q} \beta_{iq} \left(\sum_{j \in \mathcal{U}_q} |\mathbf{h}_{iq}^H \mathbf{w}_{jq}|^2 - \epsilon_{iq}\right)\!-\!\sum_{p \in \mathcal{T}} \sum_{i \in \mathcal{U}_p}
    \frac{\lambda_{ip}}{N_T } \Bigg[\frac{|\textbf{h}_{ip}^H \textbf{w}_{ip}|^2}{\gamma_{ip}} \\
    & - \sum_{j \in \mathcal{U}_p \setminus i} |\textbf{h}_{ip}^H \textbf{w}_{jp}|^2 - \sum_{q \in \mathcal{T}_i \setminus p} \tau_{iq} - \sum_{q \in \mathcal{T} \setminus \mathcal{T}_i} \epsilon_{iq} - \sigma^2\Bigg],
    \end{aligned}
\end{equation}
where the Lagrangian dual variables $\{\lambda_{ip}\}$ correspond to SINR constraints, and $\{ \alpha_{iq}, \beta_{iq}\}$ are associated with interference constraints.
Assuming $\lambda_{ip} = 0$, the KKT condition for $\textbf{w}_{ip}$ becomes
\begin{equation} \label{eq:kkt1}
    \begin{aligned}
       \Bigg(\textbf{I} + \sum_{j \in \mathcal{U}_p \setminus i} \alpha_{jp} \textbf{h}_{jp} \textbf{h}_{jp}^H + \sum_{j \notin \mathcal{U}_p}   &\beta_{jp} \textbf{h}_{jp} \textbf{h}_{jp}^H  \\
       & + \sum_{j \in \mathcal{U}_p \setminus i} \frac{\lambda_{jp}}{N_T} \textbf{h}_{jp} \textbf{h}_{jp}^H \Bigg) \textbf{w}_{ip} = 0, 
    \end{aligned}
\end{equation}
Due to the non-negativity of $\{\alpha_{iq}\}, \{\beta_{iq}\},  \{\lambda_{ip}\}$, the unique solution to (\ref{eq:kkt1}) is $\textbf{w}_{ip} = \boldsymbol{0}$, which contradicts the SINR constraints in (\ref{eq:opt22}). Therefore, $\lambda_{ip}$ must be positive. Assuming $\alpha_{iq} = 0$, the coefficient of $\tau_{iq}$ in (\ref{eq:L1}) is $\sum_{p \in \mathcal{T}_i \setminus q} \lambda_{ip}$. Since $\lambda_{ip} >0,  \forall p \in \mathcal{T}, \forall i \in \mathcal{U}_{p}$, minimizing $\epsilon_{iq}$ over $\epsilon_{iq} \geq 0$ yields $\epsilon_{iq} = 0$, which contradicts constraints in (\ref{eq:opt23}). Hence, we conclude that $\alpha_{iq}$ must be positive. Following similar statements, $\beta_{iq}$ is also positive. Based on complementary slackness conditions, constraints (\ref{eq:opt23}) and (\ref{eq:opt24}) hold, indicating Problem (\ref{eq:opt1}) and Problem (\ref{eq:opt2}) are equivalent.


Next, we derive the dual problem of Problem (\ref{eq:opt1}) and Problem (\ref{eq:opt2}). In particular, $L(\mathbf{w}_{ip}, \lambda_{ip})$ is defined by substituting $\tau_{iq}, \epsilon_{iq} $ with $\sum_{j \in \mathcal{U}_q \setminus i } |\mathbf{h}_{iq}^H \mathbf{w}_{jq}|^2$, $\sum_{j \in \mathcal{U}_q} |\mathbf{h}_{iq}^H \mathbf{w}_{jq}|^2$ in (\ref{eq:L1}). By rearranging terms in $ L(\mathbf{w}_{ip}, \lambda_{ip})$, we obtain
\begin{equation} \label{eq:L2}
    \begin{aligned}
        L(\mathbf{w}_{ip}, \lambda_{ip}) = \sum_{p \in \mathcal{T}} &\sum_{i \in \mathcal{U}_p} \Bigg[\mathbf{w}_{ip} \Bigg(\mathbf{I} - \frac{\lambda_{ip}}{N_T \gamma_{ip}} \mathbf{h}_{ip}^H h_{ip}  \\
        &+ \sum_{j \in \mathcal{U} \setminus i} \sum_{q \in \mathcal{T}_j} \frac{\lambda_{jq}}{N_T} \mathbf{h}_{jp} \mathbf{h}_{jp}^H \Bigg)\mathbf{w}_{ip}^H + \frac{\lambda_{ip}}{N_T} \sigma^2 \Bigg]. 
    \end{aligned}
\end{equation}
The dual objective function is defined as $g(\lambda_{ip}) = \min_{\textbf{w}_{ip} } L(\mathbf{w}_{ip}, \lambda_{ip})$. If the matrix within the parentheses is not semi-positive definite matrix, we have $g(\lambda_{ip}) = \infty$. Therefore, we get the Lagrangian dual problem of ($\ref{eq:opt1}$) and ($\ref{eq:opt2}$) as follows:
\begin{equation}  \label{eq:opt4}
  \begin{aligned}
    &\max_{\{\lambda_{ip}\}} \ \sum_{p \in \mathcal{T}} \sum_{i \in \mathcal{U}_p} \frac{\lambda_{ip}}{N_T} \sigma^2 \\
    &\ \text{s.t.} \quad \mathbf{I} + \sum_{j \in \mathcal{U} \setminus i} \sum_{q \in \mathcal{T}_j} \frac{\lambda_{jq}}{N_T} \mathbf{h}_{jp} \mathbf{h}_{jp}^H \succeq \frac{\lambda_{ip}}{N_T \gamma_{ip}} \textbf{h}_{ip} \textbf{h}_{ip}^H,  \\
    &\qquad \qquad \qquad \qquad \qquad \qquad \quad\forall p \in \mathcal{T}, \forall i \in \mathcal{U}_p.
  \end{aligned}
\end{equation}

\section{Solution of prime problem (\ref{eq:opt2}) and dual problem (\ref{eq:opt4})} \label{sec:solution}
We first compute the gradient of (\ref{eq:L2}) with respect to $\textbf{w}_{ip}$ and equal it to zero, yielding:
\begin{equation} \label{eq:da1}
    \boldsymbol{\Sigma}_{ip} \textbf{w}_{ip} = \frac{\lambda_{ip}}{\gamma_{ip}} \textbf{h}_{ip} \textbf{h}_{ip}^H \textbf{w}_{ip}, \quad \forall p \in \mathcal{T}, \forall i \in \mathcal{U}_p,
\end{equation}
where $\boldsymbol{\Sigma}_{ip} \triangleq N_T \mathbf{I} + \sum_{j \in \mathcal{U} \setminus i} \sum_{q \in \mathcal{T}_j} {\lambda}_{jq} \mathbf{h}_{jp} \mathbf{h}_{jp}^H$. Subsequently, by left-multiplying both sides of (\ref{eq:da1}) with $\textbf{h}_{ip}^H \boldsymbol{\Sigma}_{ip}^{-1}$, we obtain:
\begin{equation} \label{eq:da2}
    \textbf{h}_{ip}^H \textbf{w}_{ip} = \frac{\lambda_{ip}}{ \gamma_{ip}} \textbf{h}_{ip}^H \boldsymbol{\Sigma}_{ip}^{-1} \textbf{h}_{ip} \textbf{h}_{ip}^H \textbf{w}_{ip}, \quad \forall p \in \mathcal{T}, \forall i \in \mathcal{U}_p.
\end{equation}
By cancelling $\textbf{h}_{ip}^H \textbf{w}_{ip}$ from both sides of the (\ref{eq:da2}), we obtain the coupled equation set for $\{\lambda_{ip}\}$ in  (\ref{eq:lamda}).

Next, we solve the equation (\ref{eq:da1}) after obtaining the expression of $\{\lambda_{ip}\}$. It is easy to check that $\hat{\textbf{w}}_{ip} = \boldsymbol{\Sigma}_{ip}^{-1} \textbf{h}_{ip}$ is a solution of (\ref{eq:da1}). By substituting $\hat{\textbf{w}}_{ip}$ into (\ref{eq:da1}), the left-hand sides reduces to $\textbf{h}_{ip}$ while the right-hand side is expressed as $\frac{\lambda_{ip}}{\gamma_{ip}} \textbf{h}_{ip}^H \boldsymbol{\Sigma}_{ip}^{-1} \textbf{h}_{ip}^H \textbf{h}_{ip}$. Besides, (\ref{eq:da2}) implies that $1 = \frac{\lambda_{ip}}{\gamma_{ip}} \textbf{h}_{ip}^H \boldsymbol{\Sigma}_{ip}^{-1} \textbf{h}_{ip}^H$. Consequently, the right-hand side of (\ref{eq:da1}) and the left-hand side of (\ref{eq:da1}) are equivalent. Similar to \cite{wyu2007}, we set  $\textbf{w}_{ip} = \sqrt{\frac{\delta_{ip}}{N_T}} \hat{\textbf{w}}_{ip}$, where $\delta_{ip}$ serves as a scaling factor determining the power. To calculate $\delta_{ip}$, we observe that $\lambda_{ip}$ is non-zero, implying that the SINR constraints in (\ref{eq:opt22}) hold with equality due to complementary slackness conditions. By taking $\textbf{w}_{ip}$ into (\ref{eq:opt22}) and ensuring the satisfaction of SINR constraints with equality, we can get $\textbf{w}_{ip}$ by solving the linear equation for $\delta_{ip}$, as previously shown in Section \ref{sec:central}.

\section{Proof of Theorem \ref{eq:th1}} 
\label{sec:proof1}
We firstly propose a heuristic approach  for the deterministic approximation of ${\lambda_{iq}}^*$, with the assumption that the optimal Lagrangian multipliers are deterministic values, independent of channel vectors. Defining $\boldsymbol{\lambda}^* = \{\lambda_{ip}^*\}$ as the fixed point of (\ref{eq:lamda}), we establish the following equality:
\begin{equation} \label{eq:B1}
    \frac{\gamma_{ip}}{\lambda_{ip}^*} = \frac{1}{N_T} \mathbf{h}_{ip}^H(\mathbf{I} +  \sum_{j \in \mathcal{U} \setminus i} \sum_{q \in \mathcal{T}_j} \frac{\lambda_{jq}^*}{N_T} \mathbf{h}_{jp} \mathbf{h}_{jp}^H)^{-1} \mathbf{h}_{ip}. 
\end{equation}
Under the mistaken premise that  $\boldsymbol{\lambda}^*$ is pre-defined and independent of channel vectors, the trace lemma yields the following result:
\begin{equation} \label{eq:B2}
     \frac{\gamma_{ip}}{\lambda_{ip}^*} - \frac{1}{N_T} \operatorname{Tr} [\boldsymbol{\Theta}_{ip} (\mathbf{I} +  \sum_{j \in \mathcal{U} \setminus i} \sum_{q \in \mathcal{T}_j} \frac{\lambda_{jq}^*}{N_T} \mathbf{h}_{jp} \mathbf{h}_{jp}^H)^{-1} ] \to 0, \; a.s.
\end{equation}
Due to the finiteness of $T_i$, the cardinality of the set of BSs serving $i$, the rank-one perturbation lemma allows adding $\sum_{q\in \mathcal{T}_i} \frac{\lambda_{iq}^*}{N_T} \textbf{h}_{ip} \textbf{h}_{ip}^H$ to the matrix within parentheses in (\ref{eq:B2}) without altering the limitation. Then we have:
\begin{equation} \label{eq:B3}
     \frac{\gamma_{ip}}{\lambda_{ip}^*} - \frac{1}{N_T} \operatorname{Tr} [\boldsymbol{\Theta}_{ip} (\mathbf{I} +  \sum_{j \in \mathcal{U}} (\sum_{q \in \mathcal{T}_j} \frac{\lambda_{jq}^*}{N_T} )\mathbf{h}_{jp} \mathbf{h}_{jp}^H)^{-1} ] \to 0, \; a.s.
\end{equation}
Denote the second term of (\ref{eq:B3}) as $m_{ip}^*$, then $m_{ip}^*$ converges to $\bar{m}_{ip}^*$ almost surely according to DE theory \cite{RMT2011}.  $\bar{m}_{ip}^*$ is given by the fixed point of the system equation as follows:
\begin{equation}  \
    \bar{m}_{ip}^* = \operatorname{Tr}(\boldsymbol{\Theta}_{ip} (\sum_{j \in \mathcal{U}} \frac{(\sum_{q \in \mathcal{T}_j } {\lambda}_{jq}^* ) \boldsymbol{\Theta}_{jp}}{1+ (\sum_{q \in \mathcal{T}_j } {\lambda}_{jq}^* ) \bar{m}_{jp}^*} + N_T \mathbf{I})^{-1}).
\end{equation}
However, the optimal Lagrangian multipliers $\{\lambda_{ip}^*\}$ are functions of channel vectors instead of fixed values. Then, in order to complete the proof, we need to show $\{\lambda_{ip}^*\} $ is close to $\frac{\gamma_{ip}}{\bar{m}_{ip}^*}$ in the large system limitation. This is analogous to the proof presented in \cite{has2019}, and we will not elaborate further.

\section{Proof of Theorem \ref{eq:th2}} 
\label{sec:proof2}
Given the DE of the optimal Lagrangian multipliers as $\{\bar{\lambda}_{ip}\}$, we can deduce the corresponding DE of the entries $F^{pq}_{ij}$. Firstly, we rewrite the diagonal elements $F^{pp}_{ii}$ in coupling matrix (\ref{eq:F2}) as $F^{pp}_{ii} = \frac{1}{\gamma_{ip}} |\frac{1}{N_T} \textbf{h}_{ip}^H \boldsymbol{\Theta}_{p}^{\setminus i} \textbf{h}_{ip}|^2$ with the definition $\boldsymbol{\Sigma}_{p}^{\setminus i_1,\cdots,i_k} = (\textbf{I} + \sum_{j\in \mathcal{U} \setminus \{i_1,\cdots,i_k\} } \sum_{q \in \mathcal{T}_j} \frac{\bar{\lambda}_{jq}}{N_T} \textbf{h}_{jp} \textbf{h}_{jp}^H)^{-1}$. The trace lemma yields $\frac{1}{N_T} \textbf{h}_{ip}^H \boldsymbol{\Theta}_{p}^{\setminus i} \textbf{h}_{ip} - \frac{1}{N_T} \operatorname{Tr} (\boldsymbol{\Theta}_{ip} \boldsymbol{\Sigma}_{p}^{\setminus i})$ converges to zero almost surely. Applying the rank-one perturbation lemma, we can replace the trace term with $\operatorname{Tr} (\boldsymbol{\Theta}_{ip} \boldsymbol{\Sigma}_i)$ without limitation change where $\boldsymbol{\Sigma}_i = (\textbf{I} + \sum_{j \in \mathcal{U}} \sum_{q \in \mathcal{T}_j} \frac{\bar{\lambda}_{jq}}{N_T} \textbf{h}_{jp} \textbf{h}_{jp}^H)^{-1}$. The trace term obtained is equal to $\bar{m}_{ip}$ as defined in (\ref{eq:dem}). This implies that $F^{pp}_{ii} - \frac{1}{\gamma_{ip}} |\bar{m}_{ip}|^2 \to 0$ almost surely. 

The non-zero non-diagonal elements of the coupling matrix $F^{pq}_{ij} = -\frac{1}{N^2} \textbf{h}_{iq}^H \boldsymbol{\Sigma}_{q}^{\setminus j} \textbf{h}_{jq} \textbf{h}_{jq}^H \boldsymbol{\Sigma}_q^{\setminus j} \textbf{h}_{iq}$ can be rewritten as:
\begin{equation} \label{eq:deFip}
    F_{pq}^{ij} = -\frac{1}{N_{T}^2} \frac{\textbf{h}_{iq}^H \boldsymbol{\Sigma}_q^{\setminus i,j} \textbf{h}_{jq} \textbf{h}_{jq}^H \boldsymbol{\Sigma}_q^{\setminus i,j} \textbf{h}_{jq} \textbf{h}_{iq}}{\left[1 + \left(\frac{1}{N_T} \sum_{r \in \mathcal{T}_i} \bar{\lambda}_{ir}\right) \textbf{h}_{iq}^H \boldsymbol{\Sigma}_q^{\setminus i,j} \textbf{h}_{iq}\right]^2},
\end{equation}
by using matrix inversion equality $\textbf{q}^H (\textbf{B} + \tau \textbf{q} \textbf{q}^H)^{-1} = \frac{1}{1 + \tau \textbf{q}^H \textbf{B}^{-1} \textbf{q}} \textbf{q}^H \textbf{B}^{-1}$ where $\textbf{q} \in \mathbb{C}^{N_T}$ and $B \in \mathbb{C}^{N_T \times N_T}$ is positive definite matrix. Since $T_i$ and $T_j$ are constant, we can apply trace lemma and rank-one lemma to the denominator in (\ref{eq:deFip}) similar to previous analysis, getting the following result:
\begin{equation}
    \left[1\!+\!(\frac{1}{N_T} \sum_{r \in \mathcal{T}_i} \bar{\lambda}_{ir}) \textbf{h}_{iq}^H \boldsymbol{\Sigma}_{q}^{\setminus i,j} \textbf{h}_{iq}]^2\!-\![1 + (\sum_{r \in \mathcal{T}_i} \bar{\lambda}_{ir}) \bar{m_{iq}}\right]^2\!\to\!0, \ a.s.
\end{equation}
We continue by applying the trace lemma on the numerator of (\ref{eq:deFip}), yielding the following result:
\begin{equation} \label{eq:deFip2}
    \begin{aligned}
        \frac{1}{N_T^2} \textbf{h}_{iq}^H \boldsymbol{\Sigma}_q^{\setminus i,j} \textbf{h}_{jq} \textbf{h}_{jq}^H \boldsymbol{\Sigma}_q^{\setminus i,j} \textbf{h}_{iq} - \frac{1}{N_T^2} \operatorname{Tr} (\boldsymbol{\Theta}_{iq}   \boldsymbol{\Sigma}_q^{\setminus i,j}   & \textbf{h}_{jq} \textbf{h}_{jq}^H  \boldsymbol{\Sigma}_q^{\setminus i,j}) \\  &\to 0, \quad a.s. 
    \end{aligned}
\end{equation}
Rearranging terms within the trace operator in (\ref{eq:deFip2}) and subsequently reapplying the trace lemma, we get
\begin{equation}
    \textbf{h}_{jq}^H \boldsymbol{\Sigma}_q^{\setminus i,j} \frac{\boldsymbol{\Theta}_{iq}}{N_T^2} \boldsymbol{\Sigma}_q^{\setminus i,j} \textbf{h}_{jq} - \operatorname{Tr} \left(\boldsymbol{\Theta}_{jq} \boldsymbol{\Sigma}_q^{\setminus i,j} \frac{\boldsymbol{\Theta}_{iq}}{N_T^2} \boldsymbol{\Sigma}_q^{\setminus i,j}\right) \to 0, \quad a.s.
\end{equation}
Then we can add the excluded $i$-th term $\sum_{q \in \mathcal{T}_i} \frac{\bar{\lambda}_{jq}}{N_T} \textbf{h}_{jp} \textbf{h}_{jp}^H$  and $j$-th term $\sum_{q \in \mathcal{T}_j} \frac{\bar{\lambda}_{jq}}{N_T} \textbf{h}_{jp} \textbf{h}_{jp}^H$  to $\boldsymbol{\Sigma}_q^{\setminus i,j}$ according to the rank-one lemma, yielding
\begin{equation} \label{eq:deFip3}
    (-F_{ij}^{pq}) - \frac{\frac{1}{N_T^2} \operatorname{Tr} (\boldsymbol{\Theta}_{jq} \boldsymbol{\Sigma}_q \boldsymbol{\Theta}_{iq} \boldsymbol{\Sigma}_q)}{[1 + (\sum_{r \in \mathcal{T}_i} \bar{\lambda}_{ir} ) \bar{m}_{iq}]^2} \to 0, \quad a.s.
\end{equation}
Using matrix derivative formula $\partial \mathbf{Z}^{-1} / \partial x=-\mathbf{Z}^{-1}(\partial \mathbf{Z} / \partial x) \mathbf{Z}^{-1}$ where $\mathbf{Z}$ is a matrix function of variable $x$, we can express the nominator in (\ref{eq:deFip3}) as follows:
\begin{equation}
    \frac{1}{N_T^2} \operatorname{Tr}(\boldsymbol{\Theta}_{jq} \boldsymbol{\Sigma}_q \boldsymbol{\Theta}_{iq} \boldsymbol{\Sigma}_q) = \frac{\partial}{\partial x} m_{j,i,q}(z,x)|_{x=0, z=-1},
\end{equation}
with the definition $m_{j,i,q}(z,x) = \frac{1}{N_T} \operatorname{Tr}[\boldsymbol{\Theta}_{jq}(\boldsymbol{\Sigma}_q - z \textbf{I} - x \boldsymbol{\Theta}_{iq})^{-1}]$. Therefore, the final step entails computing the DE of $m_{j,i,q}^{\prime}$, i.e., the partial derivative of $m_{j,i,q}(z,x)$ with respect to $x$ at $x=0$ and $z=-1$. According to DE theory \cite{RMT2011}, the DE of $m_{j,i,q}(z,x)$ is given by $\bar{m}_{j,i,q}(z.x) = \frac{1}{N_T} \operatorname{Tr} [\boldsymbol{\Theta}_{jq} \mathbf{T}_{q,i}(z,x)]$ where
\begin{equation} \label{eq:T1}
    \mathbf{T}_{b,i}(z,x) = \left[\frac{1}{N_T} \sum_{k \in \mathcal{U}} \frac{(\sum_{r \in \mathcal{T}_k} \bar{\lambda}_{kr}) \boldsymbol{\Theta}_{kb} }{1 + (\sum_{r \in \mathcal{T}_k} \bar{\lambda}_{kr}) \bar{m}_{k,i,b}(z,x)} - x \boldsymbol{\Theta}_{ib} - z \textbf{I}\right]^{-1}.
\end{equation}
By taking the derivative of ($\ref{eq:T1}$) with respect to $x$ at $x=0$ and $z=-1$, we have 
\begin{equation} \label{eq:T2}
    \mathbf{T}_{b,i}^{\prime} = \mathbf{T}_b \left(\frac{1}{N_T} \sum_{k \in \mathcal{U}} \frac{(\sum_{r \in \mathcal{T}_k} \bar{\lambda}_{kr})^2 \boldsymbol{\Theta}_{kb} \bar{m}_{k,i,b}^{\prime}}{[1 + (\sum_{r \in \mathcal{T}_k} \bar{\lambda}_{kr}) \bar{m}_{kb}]^2} + \boldsymbol{\Theta}_{ib}\right) \mathbf{T}_b,
\end{equation}
where $\mathbf{T}_b = \mathbf{T}_{b,i}(-1,0), \mathbf{T}_{b,i}^{\prime} = \frac{\partial}{\partial x} \mathbf{T}_{b,i}(z,x)|_{z=-1, x=0}$, and $\bar{m}_{kb} = \bar{m}_{k,i,b}(-1,0)$. Then with the equality $\bar{m}_{j,i,q}^{\prime} = \frac{1}{N_T} \operatorname{Tr} (\boldsymbol{\Theta}_{jq} \textbf{T}_{q,i}^{\prime})$, we get a linear system equation of $\bar{m}_{j,i,q}^{\prime}$ and the solution is shown in Theorem \ref{eq:th2}.

\section{Derivations of the Closed-form Solution for Problem (\ref{eq:opt7}) \label{appendixupdateA}}
In order to solve Problem (\ref{eq:opt7}), we first write the Lagrangian function of this problem as follows:
\begin{equation} \label{ea:La1}
    \begin{aligned}
         \mathcal{L} &\!=\!\sum_{j\in \mathcal{U}_{p}}|A_{i,j}-\mathbf{h}_{ip}^{H}\mathbf{w}_{jp}^{(l,m)}+ \lambda_{i,j}^{(l,m)}|^{2} + \beta\left[\sum_{j \in \mathcal{U}_p \backslash i}\left|A_{i, j}\right|^{2}\!-\!\tau_{ip} \right] \\
        &+\zeta\Bigg[\gamma_{ip}\left(\sum_{j \in \mathcal{U}_p \setminus i} |A_{i,j}|^2 + \sum_{q \in \mathcal{T}_i \setminus p} \tau_{iq} + \sum_{q \in \mathcal{T} \setminus \mathcal{T}_i} \epsilon_{iq} + \sigma^2\right) \\
        &-2\Re{\{(\mathbf{w}_{ip,c}^{(l)})^{H}\mathbf{h}_{ip}A_{i,i}\}}+|\mathbf{h}_{ip}^{H}\mathbf{w}_{ip}^{(l)}|^2\Bigg] \\
        & = \sum_{j\in \mathcal{U}_{p}}\!|A_{i,j}-\mathbf{h}_{ip}^{H}\mathbf{w}_{jp}^{(l,m)}+ \lambda_{i,j}^{(l,m)}|^{2}\!+\!(\zeta \gamma_{ip}\!+\!\beta)\!\sum_{j \in \mathcal{U}_{p} \backslash i} |A_{i,j}|^2  \\
        & + \zeta [\gamma_{ip} (\sum_{q \in \mathcal{T}_i \setminus p} \tau_{iq} + \sum_{q \in \mathcal{T} \setminus \mathcal{T}_i} \epsilon_{iq} + \sigma^2) + |\mathbf{h}_{ip}^{H}\mathbf{w}_{ip}^{(l)}|^2 ] - \beta \tau_{ip} \\
        &-2 \zeta \Re{\{(\mathbf{w}_{ip,c}^{(l)})^{H}\mathbf{h}_{ip}A_{i,i}\}},
    \end{aligned} 
\end{equation}
where $\zeta>0, \beta>0 $  is the dual variables associated with the inequality constraint (\ref{eq:opt72}) and (\ref{eq:opt73}), respectively. The derivative of $\mathcal{L}$ with respect to $A_{i,j}$ :
\begin{equation}
    \frac{\partial \mathcal{L}}{\partial A_{i, j}}=\left\{\begin{array}{llll}
2\left(A_{i, j}-\mathbf{h}_{ip}^{H} \mathbf{w}_{jp}^{(l,m)}+\lambda_{i, j}^{(l,m)}\right)+2\left(\zeta \gamma_{ip}+\beta\right) A_{i, j},\\ 
\quad \quad \quad \quad \quad \quad \quad \quad \quad \quad \quad \quad \quad \quad \quad \quad \text { if } j \neq i, \\
2\left(A_{i, j}-\mathbf{h}_{ip}^{H} \mathbf{w}_{jp}^{(l,m)}+\lambda_{i,j}^{(l,m)}\right)-2 \zeta \mathbf{h}_{ip}^{H} \mathbf{w}_{ip,c}^{(l)},\\ 
\quad \quad \quad \quad \quad \quad \quad \quad \quad \quad \quad \quad \quad \quad \quad \quad \text { if } j=i.
\end{array}\right.
\end{equation}
By setting the derivatives to zero, the optimal primal variables $A_{i,j}$ can be expressed as:
\begin{equation} \label{eq:Solu1}
    A_{i, j}=\left\{\begin{array}{ll}
\frac{\mathbf{h}_{ip}^{H} \mathbf{w}_{jp}^{(l,m)}-\lambda_{i, j}^{(l,m)}}{1+\zeta \gamma_{ip}+\beta}, & \text { if } j \neq i \\
\zeta^{(l,m)} \mathbf{h}_{ip}^{H} \mathbf{w}_{ip,c}^{(l)}+\mathbf{h}_{ip}^{H} \mathbf{w}_{jp}^{(l,m)}-\lambda_{i,j}^{(l,m)}, & \text { if } j=i
\end{array}\right.
\end{equation}
Substituting (\ref{eq:Solu1}) into (\ref{eq:opt72}),(\ref{eq:opt73}) and and solving jointly the obtained equations, we derive 
\begin{align}
A_{i,j}^{(l,m+1)}= \begin{cases}\frac{\sqrt{\tau_{ip}}\left(\mathbf{h}_{ip}^{H}\mathbf{w}_{jp}^{(l,m)}-\lambda_{i,j}^{(l,m)}\right)}{\sqrt{\sum_{j \in \mathcal{U}_p \setminus i}|\mathbf{h}_{ip}^{H}\mathbf{w}_{jp}^{(l,m)}-\lambda_{i,j}^{(l,m)}|^{2}}}, & \text { if } j \neq i, \\ \zeta_{ip}^{(l,m)}\mathbf{h}_{ip}^{H}\mathbf{w}_{jp,c}^{(l)} +\mathbf{h}_{i}^{H}\mathbf{w}_{jp}^{(l,m)}-\lambda_{i,j}^{(l,m)}, & \text { if } j = i,\end{cases}
\end{align}
\noindent where $\zeta_{i}^{(l,m)}=\Big(\gamma_{ip}(\tau_{ip} + \sum_{q \in \mathcal{T}_i \setminus p} \tau_{iq} + \sum_{q \in \mathcal{T} \setminus  \mathcal{T}_i} \epsilon_{iq} + \sigma^2)+|\mathbf{h}_{ip}^{H}\mathbf{w}_{ip,c}^{(l)}|^2-2\Re\{(\mathbf{w}_{ip,c}^{(l)})^{H}\mathbf{h}_{ip}(\mathbf{h}_{ip}^{H}\mathbf{w}_{ip}^{(l,m)}-\lambda_{i,i}^{(l,m)})\}\Big)\Big/\Big(2|\mathbf{h}_{ip}^{H}\mathbf{w}_{ip,c}^{(l)}|^2\Big)$. 

\ifCLASSOPTIONcaptionsoff
  \newpage
\fi

\end{document}